\documentclass{report}
\usepackage{setspace}
\usepackage{subcaption}

\usepackage{amssymb,graphicx,color}
\usepackage{amsfonts}
\usepackage{latexsym}
\usepackage{a4wide}
\usepackage{amsmath}
\usepackage[utf8]{inputenc}
\usepackage{textcomp}
\usepackage{algorithm}
\usepackage{algpseudocode}
\usepackage{tikz}
\usetikzlibrary{backgrounds}
\usepackage{caption}
\usepackage{subcaption}
\usepackage[table]{xcolor}
\usetikzlibrary{arrows.meta, positioning, shadings}
\usepackage{xspace}
\usepackage{url}
\usepackage{minted}
\usepackage{hyperref}
\usepackage{listings}
\usepackage{needspace}
\usepackage{indentfirst}

\newtheorem{theorem}{THEOREM}
\newtheorem{lemma}[theorem]{LEMMA}
\newtheorem{corollary}[theorem]{COROLLARY}

\newenvironment{proof}{
PROOF:
\begin{quotation}}{
$\Box$ \end{quotation}}

\newcommand{\lmysticeti}{\textsf{Odontoceti}\xspace}

\newenvironment{caseof}{}{\vskip.5\baselineskip}
\newcommand{\case}[2]{\vskip.5\baselineskip\par\noindent {\bfseries Case #1}:\\\noindent #2}
\newcommand{\proofpart}[2]{\vskip.5\baselineskip\par\noindent {\bfseries #1:}\\\noindent #2}

\algtext*{EndFunction}
\algtext*{EndFor}
\algtext*{EndIf}

\lstdefinelanguage{Rust}{
  keywords={abstract,alignof,as,become,box,break,const,continue,crate,do,else,enum,extern,false,final,fn,for,if,impl,in,let,loop,macro,match,mod,move,mut,offsetof,override,priv,proc,pub,pure,ref,return,Self,self,sizeof,static,struct,super,trait,true,type,typeof,unsafe,unsized,use,virtual,where,while,yield},
  keywordstyle=\color{blue}\bfseries,
  ndkeywords={bool,char,f32,f64,i8,i16,i32,i64,isize,str,u8,u16,u32,u64,usize,Option,Result,Vec,String},
  ndkeywordstyle=\color{purple}\bfseries,
  identifierstyle=\color{black},
  sensitive=true,
  comment=[l]{//},
  morecomment=[s]{/*}{*/},
  commentstyle=\color{gray}\ttfamily,
  stringstyle=\color{red}\ttfamily,
  string=[b]",
  showstringspaces=false,
}

\title{  	{ \includegraphics[scale=.5]{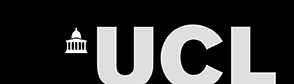}}\\
{{\Huge \lmysticeti: Ultra-Fast DAG Consensus with Two Round Commitment}}\\
		}
\date{Submission date: 8 September 2025}
\author{Preston Vander Vos\thanks{
{\bf Disclaimer:}
This report is submitted as part requirement for the MSc Information Security at UCL. It is substantially the result of my own work except where explicitly indicated in the text. The report may be freely copied and distributed provided the source is explicitly acknowledged.
}
\\ \\
MSc Information Security\\ \\
Philipp Jovanovic and Alberto Sonnino}

\begin{document}

\onehalfspacing
\maketitle
\begin{abstract}
  Users of blockchains value scalability, expecting fast confirmations and immediate transaction processing. \lmysticeti, the latest in DAG-based consensus, addresses these concerns by prioritizing low latency and high throughput, making a strategic trade-off in security by operating with a 20\% fault tolerance instead of the established 33\% level. It is the first DAG-based protocol to achieve commitment in just two communication rounds, delivering median latency of 300 milliseconds while processing 10,000 transactions per second under realistic network conditions. \lmysticeti operates with $n = 5f + 1$ validators and creates an uncertified DAG with a novel decision rule for committing blocks. The protocol includes an optimization that advances progress when participants are slow, benefiting crash fault scenarios which are more common in practice than Byzantine faults. Evaluation results demonstrate 20-25\% latency improvements compared to an existing production protocol, validating that reducing wave length from three rounds to two rounds yields meaningful performance benefits. This paper establishes the practical viability of lower fault tolerance consensus protocols for blockchains.

\end{abstract}
\tableofcontents
\setcounter{page}{1}

\chapter{Introduction}\label{chap:introduction}

Blockchains face an inherent difficulty of simultaneously achieving security, decentralization, and scalability. This fundamental design challenge, known as the blockchain trilemma, was first articulated by Vitalik Buterin \cite{trilemma}. Blockchain protocols need to decide how to balance these three properties, as improvements in one dimension come at the expense of the others. Security ensures that the system remains resilient against malicious attacks. Decentralization guarantees that no single entity controls the network, preserving its distributed nature. Scalability refers to the system's ability to quickly process many transactions.

While all three properties are crucial for a robust blockchain system, end users primarily experience scalability in daily interactions. Security and decentralization operate largely invisibly from the user's perspective. Users care most about confirmation times and fees when submitting transactions. Improved scalability delivers direct benefits across these dimensions: faster confirmations, higher throughput, and lower fees as less congested networks require fewer resources per transaction. In contrast, a slow or congested network creates friction through delays and higher costs, which overshadows even the most robust security and decentralization.

\lmysticeti, the latest innovation in DAG-based consensus protocols, addresses this user centric concern by prioritizing low latency and high throughput, thus delivering strong scalability. To accomplish this goal, \lmysticeti makes a strategic trade-off in security. Specifically, it operates with a 20\% fault tolerance instead of the established 33\% level, accepting a lower threshold in exchange for improved performance. This design choice enables \lmysticeti to achieve consensus faster than traditional protocols.

The motivation for this approach stems from the reality that blockchain systems ultimately compete with centralized payment processors like Visa and Mastercard for user adoption. These companies provide near instantaneous payment processing since they control the entire network infrastructure without the decentralization constraints that the blockchain trilemma carries. For blockchains to gain even more widespread adoption, protocols must deliver user experiences that approach the responsiveness of these centralized alternatives, even if this requires accepting certain security trade-offs.

\section{Contributions}

The primary contributions of this paper are:

\begin{itemize}
    \item \textbf{First two round DAG consensus protocol}: \lmysticeti represents the first DAG-based consensus protocol to commit blocks in just two communication rounds.
    
    \item \textbf{Sub half second latency performance}: The protocol demonstrates exceptional speed, achieving \textbf{median latency of 300 milliseconds} (0.3 seconds) while processing 10,000 transactions per second under realistic network conditions.
    
    \item \textbf{Security-performance trade-off analysis}: This paper provides a comprehensive evaluation of a consensus protocol operating under a relaxed security assumption, showing empirically that \textbf{a 20\% fault tolerance can achieve practical performance gains}.
\end{itemize}

\section{Organization}

This paper is organized as follows. Chapter~\ref{chap:background} supplies the necessary background on foundational distributed systems concepts and traces the evolution of DAG-based consensus protocols. Building on this, Chapter~\ref{chap:overview} presents the system model, threat assumptions, and high level design principles that guide \lmysticeti's architecture. Chapter~\ref{chap:protocol} gives a detailed technical specification of the protocol, including the novel decision rule, a performance optimization, and complete pseudocode. The theoretical guarantees are established in Chapter~\ref{chap:proofs}, which demonstrates formal protocol correctness through safety and liveness. Chapter~\ref{chap:implementation} transitions to practical considerations, discussing the development of \lmysticeti. Chapter~\ref{chap:evaluation} presents performance results that validate \lmysticeti's improvements against existing state of the art. Chapter~\ref{chap:related} examines the historical development and current adoption of 20\% fault tolerant consensus protocols. Finally, Chapter~\ref{chap:conclusion} summarizes the contributions and discusses future research directions.

\chapter{Background}\label{chap:background}

\section{Network Assumptions}\label{sec:network-assumptions}

Distributed systems operate in networks. The type of network impacts the designs and guarantees of any protocol the system runs. Three main models define message delivery assumptions between network participants \cite{networks}:

\begin{itemize}
    \item \textbf{Synchronous networks}: Feature message delivery within a known upper bound. The bound is called $\Delta$. Thus, a message sent at time $t$ will arrive by time $t + \Delta$. This provides predictable communication.

    \item \textbf{Asynchronous networks}: Impose no timing constraints on message delivery. Messages can be arbitrarily delayed, but they will still be delivered. A message sent at time $t$ will eventually arrive, but the delivery time remains uncertain and unbounded.

    \item \textbf{Partially synchronous networks}: Begin in an asynchronous phase but transition to synchronous behavior after an unknown Global Stabilization Time (GST). Messages sent before GST experience the unpredictable delays of asynchrony, while those sent after GST arrive within the known bound $\Delta$. Formally, a message sent at time $t$ will arrive by time $\max(t, \text{GST}) + \Delta$.
\end{itemize}

These models contain different amounts of underlying assumptions. Synchronous networks have the strongest assurances by requiring the most restrictive conditions, asynchronous networks make minimal assumptions, and partially synchronous networks combine the two.

An alternative representation (not used here) of partial synchrony assumes the network is always synchronous but with an \textit{unknown} delay bound. This is in contrast to the model above where $\Delta$ is \textit{known}. While this alternative model offers theoretical simplicity by avoiding asynchronous phases, it presents significant practical challenges. The approach requires timeout parameters to accommodate worst case network behavior rather than average case conditions, often leading to unnecessarily conservative performance during normal operation. Moreover, it has been shown that adaptive timeouts which increase delay estimates upon failures create vulnerabilities to Denial of Service attacks where adversaries can force protocols to exponentially increase the timeout values, severely worsening performance \cite{prime}.

\section{Byzantine Fault Tolerance}

One fundamental challenge of distributed systems is achieving agreement when some participants behave maliciously or unpredictably. Byzantine Fault Tolerance (BFT) addresses this challenge by enabling consensus amongst honest participants despite the presence of adversarial ones, termed \textit{Byzantine}, whose behavior remains unconstrained and arbitrary.

The theoretical foundation is from the Byzantine Generals Problem \cite{byzantine_generals}, which established that distributed systems can withstand at most one-third Byzantine nodes while maintaining correctness. Traditional BFT protocols use this insight by operating in networks of size $n \geq 3f + 1$, where $f$ represents the maximum number of Byzantine nodes. Thus, the total number of participants $n$ contains at least $2f + 1$ honest participants. These honest participants constitute the supermajority required for correctness. This configuration ensures that Byzantine nodes are always less than a third of the network ($f \leq \frac{n - 1}{3}$) and consequently, do not break the consensus protocol.

\subsection{BFT Trade-offs}\label{sec:byz_tradeoff}

Alternative protocol designs have explored different network compositions to achieve different performance characteristics. For instance, \cite{async_random}, \cite{rabin}, and more (discussed in Section~\ref{chap:conclusion}) employ networks of $n = 5f + 1$ nodes, pitting $4f + 1$ honest against $f$ Byzantine. This means the Byzantine nodes are always less than a fifth of the network ($f \leq \frac{n - 1}{5}$). The design choice reflects a fundamental trade-off between the fault tolerance threshold and the minimum communication rounds to attain consensus.

Various theoretical lower bounds have been found for consensus under different network makeups and models. In non-Byzantine settings where all consensus participants are honest, a lower bound of two message rounds has been formalized \cite{fast_no_faults}. For synchronous systems where processes can crash, the same two message rounds bound exists \cite{fischer_1982}. However, the introduction of Byzantine faults (a more formidable adversary than simple crashes) complicates the problem.

Traditional BFT protocols require $n \geq 3f + 1$ participants to tolerate up to $f$ Byzantine faults. These protocols require at least three communication rounds to reach consensus in the common case \cite{pbft}. The additional round is necessary because Byzantine participants can send conflicting information to different honest nodes, creating scenarios where honest participants cannot immediately distinguish between legitimate network delays and Byzantine deception. In the common case, the first round allows participants to propose values, the second round enables participants to observe and report what they received from others, and the third round provides sufficient information for honest nodes to identify and resolve any inconsistencies introduced by Byzantine behavior.

A major theoretical result established that protocols requiring only two communication rounds in Byzantine settings must employ at least $5f + 1$ participants to tolerate $f$ Byzantine faults \cite{fab}. The proof demonstrates that with fewer than $5f + 1$ participants, adversaries can construct scenarios where honest nodes receive identical views despite different underlying executions, violating consensus correctness. This lower bound reveals that reducing common case communication rounds from three to two requires accepting lower fault tolerance. Specifically, the fault tolerance in the network is decreased from 33\% to 20\%. The additional honest participants in a network of $5f + 1$ supply the necessary redundancy to help correct nodes distinguish between conflicting scenarios within two rounds, eliminating the need for the third round of communication that traditional 33\% BFT protocols require.

\section{Goals of Consensus}\label{sec:goals}

Consensus algorithms must maintain two fundamental properties that define their operational guarantees. These properties serve as the universal correctness standards that all consensus protocols, regardless of the specific assumptions or implementation, must hold.

\textbf{Safety} guarantees that every honest participant has the same final state of consensus as all other honest participants, making sure that conflicting decisions are never reached. This property prevents scenarios where different nodes commit contradictory outcomes. Safety must hold under all network conditions and timing scenarios, making it an invariant regardless of periods of asynchrony or network partition.

\textbf{Liveness} guarantees that the system makes progress on achieving consensus, preventing the protocol from stalling indefinitely. This property means that honest participants will eventually reach agreement and that new proposals will be processed within finite time. Unlike safety, liveness often depends on network timing assumptions, as seen in Section~\ref{sec:network-assumptions}.

These properties can be summarized as safety ensures nothing bad ever happens and liveness ensures good things eventually happen.

\section{Byzantine Atomic Broadcast}\label{sec:bab-def}

Byzantine Atomic Broadcast (BAB) is a specific instantiation of consensus as it formalizes safety and liveness. While basic consensus protocols like Byzantine Agreement \cite{byzantine_generals,pease_agreement} focus on agreement regarding a single decision, BAB addresses the harder challenge of coordinating all honest participants to agree on a sequence of multiple decisions. This ordering requirement makes BAB particularly suitable for blockchain protocols, where the sequence of blocks/messages is critical to the system.

BAB was created in order to achieve reliable message delivery with total ordering in the presence of Byzantine faults \cite{bab}. In BAB terminology, \textit{broadcast} refers to a participant proposing a message for inclusion in the final sequence, while \textit{deliver} refers to a participant outputting a message from the total ordering for local processing. Formally, a BAB protocol must satisfy four properties for any set of messages $M$ and honest participants $P$:

\begin{itemize}
    \item \textbf{Validity}: If an honest participant $p \in P$ broadcasts a message $m \in M$, then all honest participants in $P$ eventually deliver $m$.

    \item \textbf{Agreement}: If an honest participant $p \in P$ delivers a message $m \in M$, then all honest participants in $P$ eventually deliver $m$.

    \item \textbf{Integrity}: For any message $m \in M$ and honest participant $p \in P$, participant $p$ delivers $m$ at most once, and only if $m$ was previously broadcast by some participant.

    \item \textbf{Total Order}: If honest participants $p, q \in P$ both deliver messages $m_1 \in M$ and $m_2 \in M$, then $p$ and $q$ deliver these messages in the same order.
\end{itemize}

Blockchain consensus protocols are commonly analyzed against the BAB specification to test their ability to provide a total ordering.

\section{FLP Impossibility Result}\label{sec:flp}

While BFT protocols can sustain consensus properties, they face a fundamental theoretical limitation in asynchronous environments. The Fischer-Lynch-Paterson (FLP) impossibility result \cite{flp}, discovered shortly after the Byzantine Generals Problem, demonstrates that no deterministic consensus protocol can guarantee both safety and liveness in asynchronous systems where at least one process can be faulty. This impossibility was subsequently circumvented through the introduction of randomness into consensus protocols \cite{async_random}, enabling probabilistic termination during asynchrony. As a result, many asynchronous consensus protocols currently incorporate randomness as their primary means for achieving termination.

The FLP impossibility extends beyond purely asynchronous networks to impact partially synchronous ones as well. Since partially synchronous networks experience periods of asynchrony before entering synchrony, deterministic consensus protocols must rely on timeouts to ensure progress. A \textit{timeout} is when a participant waits a predefined duration before moving on to the next step in the protocol.

\section{Quorum Intersection}

A \textit{quorum} is a subset of nodes whose collective agreement is sufficient to make binding decisions for the entire network. \textit{Quorum intersection} guarantees that any two quorums must share at least one honest participant, preventing conflicting decisions from being simultaneously accepted. In BFT networks of $3f + 1$ nodes, quorums are $2f + 1$ in size. Then, any two quorums overlap by a minimum of $f + 1$ nodes. Since $f$ of the shared nodes can be malicious, there will always be at least one honest participant who witnesses both decisions.

A proof by contradiction illustrates why conflicting quorum decisions cannot coexist. Consider two conflicting quorums $Q_1$ and $Q_2$, each containing $2f + 1$ participants from a network of $3f + 1$ total participants. Since $f$ participants are dishonest, there are $f + 1$ honest participants which belong to $Q_1$ and a distinct set of $f + 1$ honest participants which belong to $Q_2$. The same honest participant would not belong to both $Q_1$ and $Q_2$ as their decisions conflict each other. This means there are $(f + 1) + (f + 1) = 2f + 2$ honest participants. With the $f$ dishonest participants, the network contains $(2f + 2) + f = 3f + 2$ participants. This is impossible as the network only has $3f + 1$ total participants. Thus, quorums $Q_1$ and $Q_2$ cannot exist and any quorum of size $2f + 1$ makes binding decisions for the network.

\section{Directed Acyclic Graphs}

A graph $G = (V, E)$ comprises a vertex set $V$ and an edge set $E$, with each edge linking two vertices. Graphs can be undirected, lacking edge orientation, or directed, where edges explicitly point from one vertex to the other.

A Directed Acyclic Graph (DAG) is a directed graph without any cycles. A cycle exists when starting at a vertex and following a path of connected edges returns to the starting vertex. A \textit{sub-DAG} represents a subset of vertices from the original DAG along with their interconnecting edges. An example of a DAG and a sub-DAG is shown on the right in Figure~\ref{fig:linear_dag}.

Graphs are useful when modelling blockchains. Blocks are the vertices and references between the blocks are the edges. The edges are directed as a block will reference (or point to) one or multiple other blocks. Blocks can only reference previously created blocks, so there are no cycles. This forms a DAG.

\section{Linear Consensus}

The emergence of blockchain technology began when Satoshi Nakamoto created Bitcoin \cite{bitcoin} in 2008. Bitcoin's revolutionary approach to solving the double spending problem without requiring a trusted central authority established the foundation for all subsequent blockchains.

Linear consensus, pioneered by Bitcoin, is a distributed agreement approach where participants converge on a single, totally ordered sequence of blocks. This forms a chain-like structure of blocks, hence the name blockchain. Bitcoin accomplished this by requiring each block to contain a hash pointer to its immediate predecessor. Thus, when miners (what nodes in Bitcoin's network are called) propose a new block, they must reference the previous block in the chain along with including new transactions, a valid Proof of Work solution, and more. The proposed block is then broadcast across the network, where other nodes verify it before accepting it as the next link in the chain.

In linear consensus, when multiple valid blocks are simultaneously proposed which reference the same previous block, a fork occurs. When this happens, different parts of the network accept different blocks as the correct next block. Honest miners will not accept both conflicting blocks so they take one or the other. Fork resolution in linear consensus systems relies on the longest chain rule. The rule handles forks by designating the chain with more blocks as the correct chain. Thus, miners continuously monitor for longer chains and will reorganize their local blockchain state to adopt any chain that is longer than what they currently know. This orphans shorter competing forks. The process ensures eventual consistency as the network converges on a single final sequence. Many blockchains after Bitcoin have adopted variations of this longest chain rule as their consensus protocol. The left of Figure~\ref{fig:linear_dag} shows a linear chain structure with a fork.

\begin{figure}[hbtp]
    \centering
    \begin{tikzpicture}[scale=1, transform shape, node distance=0.5cm and 0.95cm,
        every node/.style={draw, minimum size=1cm, circle},
        >=Stealth]
    
    \node[teal!75, thick] (d1) at (7, 0) {};
    \node[teal!75, thick] (d2) [right=of d1] {};
    \node (d3) [right=of d2] {};
    \node (d4) [right=of d3] {};
    \node[teal!75, thick] (d5) [below=of d1] {};
    \node[fill=teal!50] (d6) [below=of d3] {};
    \node (d7) [below=of d4] {};
    \node (d8) [below=of d5] {};
    \node (d9) [right=of d8] {};
    \node (d10) [below=of d7] {};
    
    \node (l4) [left=of d5] {};
    \node (l3) [left=of l4] {};
    \node (l2) [left=of l3] {};
    \node (l1) [left=of l2] {};
    \node[fill=red!50] (l5) [above=of l3] {};
    
    \draw[<-] (d1.east) -- (d9.west);
    \draw[<-] (d2.east) -- (d3.west);
    \draw[<-] (d3.east) -- (d7.west);
    \draw[<-] (d3.east) -- (d10.west);
    \draw[<-] (d5.east) -- (d9.west);
    \draw[<-] (d6.east) -- (d7.west);
    \draw[<-] (d6.east) -- (d10.west);
    \draw[<-] (d8.east) -- (d3.west);
    \draw[<-] (d9.east) -- (d10.west);
    \draw[<-] (d9.east) -- (d3.west);
    \draw[<-] (d5.east) -- (d4.west);
    \draw[<-] (d6.east) -- (d4.west);
    \draw[<-] (d8.east) -- (d9.west);
    
    \draw[<-, teal!75, thick] (d1.east) -- (d2.west);
    \draw[<-, teal!75, thick] (d1.east) -- (d6.west);
    \draw[<-, teal!75, thick] (d2.east) -- (d6.west);
    \draw[<-, teal!75, thick] (d5.east) -- (d2.west);
    
    \draw[<-] (l1.east) -- (l2.west);
    \draw[<-] (l2.east) -- (l3.west);
    \draw[<-] (l3.east) -- (l4.west);
    \draw[<-] (l2.east) -- (l5.west);
    
    \end{tikzpicture}
    \caption{An example of \textbf{linear} (left) vs \textbf{DAG-based} (right) block structures. On the left, the \colorbox{red!50}{red block}, which created a fork, will be orphaned as it is not part of the longest chain. On the right, the sub-DAG of the \colorbox{teal!50}{blue block} is \textcolor{teal!100}{highlighted}.}
    \label{fig:linear_dag}
\end{figure}
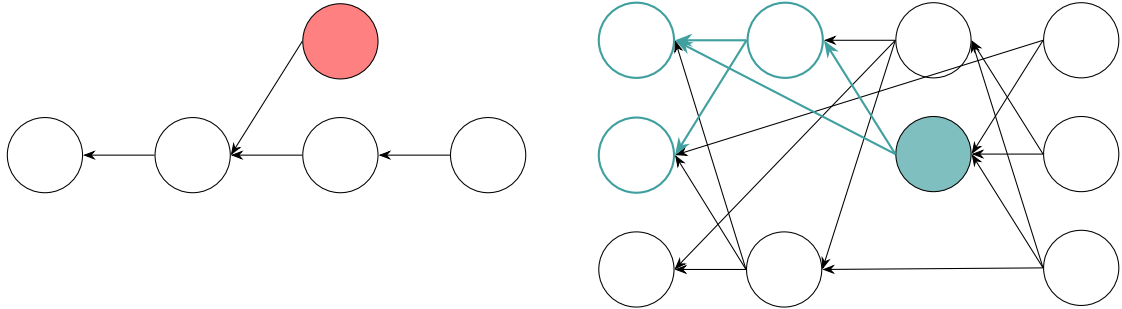

\section{DAG-Based Consensus Evolution}\label{sec:dag-evolution}

Linear consensus protocols face several inherent constraints that limit their performance. First, the sequential nature of linear blockchain construction creates a scalability bottleneck, as only one block can be added to the chain at any given time, limiting the number of transactions that can be processed. Second, fork resolution through the longest chain rule introduces a delay, as nodes must wait for other blocks to be added after theirs in order to reduce the likelihood of theirs being orphaned. Third, increasing the number of nodes in the network does not improve performance, rather it only increases competition for adding the next block in the chain.

\subsection{Origins}

The limitations of linear consensus spurred the development of alternative approaches, with Hashgraph \cite{hashgraph} emerging in 2016 as the first DAG-based consensus protocol. Rather than forcing blocks into a single chain, Hashgraph validators create blocks that reference multiple predecessors, forming a DAG structure. Operating in asynchronous network conditions, Hashgraph addressed the critical challenge of validator equivocation—where malicious participants could send conflicting blocks to different parts of the network. The protocol's solution required any block entering the total ordering to receive references from $2f + 1$ subsequent blocks created by distinct validators. Having a quorum attest to the validity of a block meant that equivocating blocks could not simultaneously be in the final output. A variation of Hashgraph is currently being used in Hedera \cite{hedera}.

\subsection{Certified DAGs and Rounds}

Building upon Hashgraph's innovation, Aleph \cite{aleph} advanced DAG-based consensus by introducing rounds and certified DAGs. Aleph tackled the equivocation problem through a different approach than Hashgraph. Instead of solely relying on subsequent block references, Aleph implemented a certification process where blocks must receive explicit approval from $2f + 1$ validators before entering the DAG.

The certification process takes three steps. Initially, a validator proposes a block and broadcasts it across the network. Other validators then verify the block's correctness and, if valid, respond with individual acknowledgments sent back to the original proposer. Once the proposer accumulates a quorum of acknowledgments from other validators, they broadcast a certificate containing proof of these approvals. The certificate of a block allows it to enter the DAG.

Aleph's other contribution was the addition of rounds, which brought logical structure to Hashgraph's previously unordered DAG construction. Rounds organize blocks into sequential stages where each round builds on the blocks of the previous rounds. This ensures consistent progress and enables validators to reason systematically about the network across time. To advance rounds, validators must collect $2f + 1$ blocks from the current round and reference them in their next proposal. 

\subsection{Leaders and Linearization}

DAG-Rider \cite{dag-rider} was the first DAG-based consensus protocol to use leaders, adapting the coordination mechanism that has been utilized in distributed systems for decades \cite{viewstamped_replication}. The protocol designates specific validators as leaders for predetermined periods, with these leaders playing a key role in consensus decisions. DAG-Rider's leader rotation emphasizes fairness by uniformly assigning every validator the same amount of leader spots.

The introduction of leaders necessitated a new approach to block ordering within the DAG. Previously, all blocks were treated equally. However, now with leaders, whose blocks drive the decision making, a distinction is needed as non-leader blocks are less important. At a high level, the consensus algorithm now focuses on agreement regarding whether leaders successfully proposed blocks in their assigned rounds or not. This creates a sequence of committed leader blocks across rounds. Once this leader sequence is found, the sub-DAGs of each committed leader block are processed through some deterministic ordering algorithm to put them in their final positions in the total ordering. Only the unique parts of each sub-DAG are included in this ordering to eliminate duplicates. This procedure of processing all the blocks after the committed leader sequence is called \textit{linearization}.

\subsection{Zero Overhead and Waves}

In combination with using leaders, DAG-Rider scored a breakthrough by eliminating extra communication for block ordering decisions. Hashgraph and Aleph required validators to exchange additional messages beyond block proposals to establish consensus on the total ordering. In contrast, DAG-Rider's commitment rule operates purely through local DAG analysis. Each validator can independently examine the DAG structure and derive the same total ordering without any supplementary communication, significantly enhancing protocol efficiency. This \textit{zero overhead} rule removes steps to achieve consensus, resulting in quicker decisions.

The approach also reduces the protocol's message complexity in the decision making phase to a single message type: blocks themselves. While the DAG construction process itself requires multiple message types for block certification, the consensus decision making phase operates exclusively through blocks that reference other blocks, eliminating the need for separate voting or commitment messages after the certified DAG is constructed. This design choice of minimizing message types was explored earlier in a Casper FFG \cite{casper_ffg} variant that finds consensus through a single message type \cite{casper_one_message}. By relying entirely on the DAG's structure for decisions, protocol implementations can be simpler.

To execute the consensus protocol, DAG-Rider organizes rounds into \textit{waves}. Waves are collections of consecutive rounds containing a single designated leader in the first round. For DAG-Rider waves are four rounds in length. Having waves provides a framework for applying the zero overhead decision rule, as the protocol evaluates whether each leader's block should enter the total ordering based on the DAG structure within and across waves.

DAG-Rider's decision rule allows leader blocks to be committed in two different ways. A direct commitment occurs if a quorum of validators ($2f + 1$) in the wave's final round maintain strong paths to the leader's block. A strong path exists when the validator and leader are connected with intermediate blocks in every round between them. This is in contrast to a weak path where the validator and leader are connected but at least one round in between them is skipped over, due to a block referencing a block in a round earlier than the directly prior round. The other way a leader block can be committed is indirectly. Indirect commitments are when a later leader block is committed and there is a strong path to a previously undecided leader block. The undecided leader will now be committed.

\subsection{Garbage Collection}

The difficulty of managing ever growing protocol state has been recognized since early consensus work, with PBFT \cite{pbft} implementing checkpoints to enable garbage collection of processed messages and preserve bounded storage requirements. Along these lines, Tusk \cite{tusk} addressed the analogous memory challenge in DAG-based consensus by introducing garbage collection capabilities. The protocol accomplished this by implementing a simple structural constraint: blocks can only reference other blocks from the immediately preceding round. This eliminated the arbitrary historical references that characterized earlier protocols like Hashgraph and DAG-Rider. The rule is vital for practical deployments, as it empowers validators to safely discard finalized portions of the DAG after consensus decisions are completed. Without the restriction, previous DAG protocols required indefinite storage of the entire DAG history due to the possibility that a block could link back to the genesis round, causing storage requirements to grow without limit.

\subsection{Pipelining}

The concept of pipelining in consensus protocols, running multiple instances in parallel, originated with early work on overlapping consensus execution. Multi-Paxos \cite{paxos_made_simple} started this by innovating to handle running multiple Paxos instances simultaneously, while Fast Paxos \cite{fast_paxos} improved on their parallel execution technique. This foundation demonstrated how overlapping consensus rounds could improve throughput by eliminating the sequential bottleneck in single instance algorithms.

Tusk adapted pipelining to DAG-based consensus through its wave architecture. The protocol employs three round waves with each wave's initial round coinciding with the preceding wave's final round. With this structural change multiple waves run simultaneously, translating the parallel execution benefits to the DAG setting.

\subsection{From Asynchronous to Partially Synchronous}

The DAG-based consensus protocols discussed thus far all operated in the asynchronous network setting, requiring randomness to circumvent the FLP impossibility result. Meanwhile, the partially synchronous setting had been dominated by non-DAG leader-based protocols such as PBFT \cite{pbft}, HotStuff \cite{hotstuff}, Tendermint \cite{tendermint}, and SBFT \cite{sbft}, which ascertained progress through timeouts but suffered from expensive recovery procedures when leaders failed.

Bullshark \cite{bullshark} bridged this gap by becoming the first DAG-based consensus protocol designed for a partially synchronous network\footnote{Bullshark also presents an asynchronous variant, but this paper focuses exclusively on the partially synchronous version.}. It was able to combine the structural advantages of DAGs with the progress guarantees of partial synchrony. The protocol's central innovation was its improved handling of faulty leaders compared to existing partially synchronous protocols. The traditional approaches require costly view-changes and view-synchronizations when leaders become unresponsive. These involve additional communication of either $O(n)$ or $O(n^2)$ messages which slows down consensus.

Bullshark eliminates expensive recovery fallbacks through a streamlined approach that incorporates timeouts directly into the DAG structure. The protocol maintains progress by having validators advance rounds dependent on block availability rather than burdensome leader coordination. When functioning properly, validators progress after collecting the leader's block and $2f$ other blocks from the current round for a total of $2f + 1$ blocks. However, if a leader becomes unresponsive or malicious, validators wait for a timeout and continue to the next round when it runs out. This straightforward design restricts the impact of faulty leaders to a single timeout delay per round, avoiding the compounding delays and extra communications that plague traditional recovery processes in partially synchronous protocols.

\subsection{Mysticeti}

Mysticeti\footnote{This paper uses Mysticeti synonymously with Mysticeti-C. Mysticeti-FPC is not covered as it is a consensusless protocol.} \cite{mysticeti} represents the culmination of DAG-based consensus evolution. Operating as a BFT consensus protocol in the partially synchronous setting, it utilizes leaders, pipelining, zero communication overhead, and garbage collection.

Like Hashgraph, Mysticeti constructs an uncertified DAG. Unlike certified approaches that require multiple communication steps (propose, acknowledge, certify) per block, uncertified DAGs allow validators to propose blocks through a single broadcast. Thus, Mysticeti eliminates all auxiliary message types. The protocol operates with only block proposals; there are no separate certification messages. This choice reintroduces the equivocation challenge where malicious validators may create conflicting blocks within the same round. To handle this trade-off, Mysticeti develops a commitment rule that keeps safety without requiring certifications.

Zero overhead commitment rules were simplified throughout the earlier protocols. They began with DAG-Rider evaluating support for the leader by seeing how many connections there were from the wave's final round back to the leader block. Tusk simplified this process by only looking at the support in the round immediately following the leader rather than waiting for the wave to end. Under Tusk's approach, leader blocks were committed after receiving $f + 1$ supports in the next round, reducing commitment latency while preserving safety.

Without certifying blocks, Mysticeti had to create a stronger commitment rule. It did this by synthesizing the approaches of DAG-Rider and Tusk. Mysticeti organizes the DAG into waves of three rounds with leaders in the first round. A leader's block is committed through a two-stage process. First, the leader's block must receive $2f + 1$ votes in the second round, echoing Tusk's immediate next round evaluation. Second, in the third round, $2f + 1$ blocks must support each of the $2f + 1$ second round blocks that voted for the leader, reflecting DAG-Rider's final round verification strategy. This dual-layer commitment rule provides the necessary checks to prevent equivocations in the uncertified DAG.

Mysticeti also incorporates an indirect decision rule by which earlier undecided leaders can be committed through a connection to a subsequently committed leader. The difference in Mysticeti from DAG-Rider's approach (which was mirrored in Tusk) requires that the connection between the future committed leader and the leader in question must pass through one of the third-round blocks that supports the $2f + 1$ second round blocks which voted for the original leader.

Unlike its predecessors which have only one leader per wave, Mysticeti has capacity for multiple leaders per wave, increasing throughput through parallel leader block proposals that each undergo independent verification. These combined innovations make Mysticeti a practical, high-performance consensus protocol. This is demonstrated by its current deployment in the Sui blockchain \cite{suiwebsite}.

Developed in roughly the same timeframe, Cordial Miners \cite{cordial_miners} explored similar ideas as Mysticeti. While both protocols share their approach to uncertified DAGs, Mysticeti distinguishes itself through creating a protocol implementation, including an indirect decision rule, pipelining waves together, and supporting multiple leaders per wave.

\subsection{Alternative Directions}

Mysticeti made specific design choices in its approach to DAG-based consensus, but many other protocols have explored alternative directions.

Shoal \cite{shoal} improves efficiency through a leader reputation scoring system, where validators assess a leader's performance based on their behavior. Nodes who consistently broadcast valid blocks and sustain honest network participation earn higher reputation scores, while those who exhibit faulty or malicious behavior earn low scores. A validator's reputation score influences their likelihood of being selected as a leader, creating economic incentives for honest behavior. Since the protocol favors reliable validators to be leaders, validators who have crashed are avoided. This reduces the probability that nodes are required to wait the full timeout and will speed up consensus.

While Mysticeti brought back uncertified DAGs to simplify communications, other protocols continue to build certified DAGs. Shoal++ \cite{shoalpp}, built on top of Shoal, kept the certified DAG structure to a degree while introducing multiple optimizations to reduce latency. The protocol commits faster by allowing some decisions to occur on uncertified proposals rather than waiting for full certification and implementing anchor pipelining where all nodes can serve as anchors to eliminate delays in finding one. Similarly, Sailfish \cite{sailfish} explores improvements for certified DAGs by introducing explicit no-vote messages that equip validators to signal their inability to support specific blocks. The protocol innovatively uses these no-vote messages to accelerate consensus decisions through expedited resolution when validators will not support blocks.

Ebb-and-flow consensus protocols \cite{sleepy,ebb_flow}, such as Slipstream \cite{slipstream}, represent a fundamentally different approach to network assumptions. They distinguish between participants being Byzantine and temporarily unavailable. This is in contrast to traditional partially synchronous protocols which treat crashed nodes as Byzantine faults, limiting the total number of non-participating nodes to at most one-third of the network. Ebb-and-flow protocols recognize that temporary unavailability is different from malicious behavior, enabling systems to tolerate higher rates of non-participation. Slipstream handles this by tolerating up to 50\% of awake nodes being Byzantine in its ``sleepy'' mode, compared to the 33\% fault tolerance in its ``synchronous'' mode. This dual approach sustains liveness under high unavailability and a total ordering that ensures safety under traditional Byzantine assumptions. Their key motivation is that real world networks often experience temporary partitions and node unavailability that should not be conflated with malicious attacks. These network assumptions empower ebb-and-flow protocols to maintain progress in scenarios where all other BFT protocols would halt.

\section{Advantages of DAG-Based Consensus}

DAG-based consensus protocols directly address the shortcomings of linear consensus through their graph structure. The ability to reference multiple predecessors enables parallel block creation, substantially increasing the amount of transactions that can be processed as validators can propose blocks simultaneously. Speed improvements come from not needing to wait a long time to see if a block is orphaned, since the DAG structure provides stronger consistency guarantees through its interconnected references. Additionally, computational efficiency increases because many blocks per round can coexist within the DAG and contribute to consensus, eliminating the waste associated with orphaned blocks in linear chains. Finally, adding more validators improves performance as more blocks are able to be created in every round. These combined advantages establish DAG-based consensus as a more scalable approach to blockchains than their linear counterparts.

\chapter{Overview}\label{chap:overview}

\section{Setting}

\lmysticeti operates within a Proof of Stake blockchain setting where a committee of validators is responsible for consensus. In Proof of Stake systems, validators are selected to be in the committee depending on their stake in the network, creating economic incentives for honest behavior since malicious behavior risks the loss of staked assets. The blockchain is organized into epochs, discrete time periods during which the committee remains fixed. Before each epoch, a selection process chooses validators to form the committee from their proportions of overall stake and other protocol-specific criteria. This committee structure provides stability during consensus execution while allowing for validator set updates between epochs to accommodate changes in stakes or participation. This framework is similar to most current blockchain consensus protocols. Throughout this paper, it is assumed that each validator in the committee has equal stake. The \lmysticeti protocol is presented during the execution of a single epoch. For simplicity, the committee concept is subsequently ignored and validators are referred to in general with the assumption that they are part of the current committee; validators outside the committee do not participate in consensus and are disregarded.

The complete transaction lifecycle involves multiple stages from user creation to final execution. Users create and digitally sign transactions, which are then broadcast to the validator network through direct connections or intermediary services. Validators receive these transactions, perform checks including signature verification and balance confirmation, and add valid transactions to their local mempools. During block production, validators select transactions from their mempools based on fees and protocol rules and bundle them into blocks. The blocks are then propagated through the network where they undergo the \lmysticeti consensus process to determine their inclusion and order in the chain. Once a block achieves consensus and is committed to the chain, the transactions within that block are considered finalized and their state changes are applied to the global blockchain state. Individual transaction finality is entirely dependent on block finality. Transactions cannot be executed or considered final until the block they are contained in successfully completes the consensus process. Therefore, this paper focuses exclusively on block-level consensus, treating blocks as having arbitrary payloads rather than examining individual transactions.

\section{Threat Model, Goals, and Assumptions}

\lmysticeti operates in the partially synchronous network model described in Section~\ref{sec:network-assumptions}, where messages are eventually delivered within bound $\Delta$ after GST. This prioritizes practical deployability as it more closely aligns with real world network behavior, where communication typically exhibits bounded delays under normal conditions but may experience unbounded delays during periods of network congestion or attack. The partially synchronous model enables \lmysticeti to optimize for common case performance by setting timeout parameters according to reasonable network assumptions while preserving safety guarantees even when these assumptions are violated during occasional disruptions.

The protocol employs a BFT design with $n = 5f + 1$ validators. A computationally bounded adversary can corrupt up to $f$ validators who may deviate arbitrarily from the protocol, while all other validators follow the protocol faithfully. Thus, there are $4f + 1$ honest and $f$ Byzantine participants. Due to the adversary being computationally bounded, standard cryptographic assumptions hold, ensuring the security of hash functions, digital signatures, and other primitives. Links between honest validators are reliable and authenticated, meaning all messages among honest parties eventually arrive and receivers can verify sender identity.

Under these network and adversary assumptions, the protocol maintains the fundamental consensus properties of safety and liveness outlined in Section~\ref{sec:goals}, proved in Section~\ref{chap:proofs}, while solving BAB as defined in Section~\ref{sec:bab-def}. BAB is selected in line with prior work \cite{dag-rider,tusk,bullshark,shoal,mysticeti,mahimahi}.

\section{Intuition Behind \lmysticeti Design}

The evolution of DAG-based consensus protocols, as described in Section~\ref{sec:dag-evolution}, has progressively reduced commitment latency through various structural and algorithmic improvements. Among these developments, uncertified DAGs have emerged as a promising approach, exemplified by protocols like Mysticeti \cite{mysticeti} and Mahi-Mahi \cite{mahimahi}. This trend is further validated by recent research in certified DAG protocols. Shoal++ \cite{shoalpp} introduced an optimization that treats certain blocks in an uncertified manner, indicating that they recognize the benefits of reducing the demands of certifications.

Certified DAG protocols require three message delays (propose, acknowledge, certify) before consensus can even evaluate a block for commitment. In contrast, uncertified DAGs eliminate this certification bottleneck. Mysticeti can commit blocks in three message delays total under partially synchronous conditions, while Mahi-Mahi achieves commitment in four or five message delays in fully asynchronous conditions, both representing substantial improvements over their certified counterparts.

\lmysticeti ventures into underexplored territory by adopting the network size $n = 5f + 1$ in the blockchain setting. This departs from the well-established $3f + 1$ size. While theoretical results, discussed in Section~\ref{sec:byz_tradeoff}, demonstrate that networks with $5f + 1$ participants have a lower bound of two communication rounds for consensus, compared to the three round lower bound in traditional systems, the implications of this transition remain largely unknown. \lmysticeti aims to explore the question of whether the expectation of reduced latency justifies the trade-off of accepting a lower fault tolerance threshold.

\section{Structure of \lmysticeti DAG}

\lmysticeti operates over a series of consecutive rounds that are grouped into waves. Each round sees every honest validator create a single block while Byzantine validators may not create a block or can equivocate by sending multiple blocks. At all times, validators collect transactions from clients and blocks from peers, integrating both into their block proposals. Each new block contains references to blocks from the directly preceding round and includes new transactions not yet included in prior blocks. If a client notices that their transaction is taking a while to be included in a block, they will resubmit it to a different validator.

\subsection{Blocks}

Formally, all blocks must include:

\begin{itemize}
    \item The author of the block along with their digital signature over the block contents
    
    \item A round number
    
    \item A list of transactions
    
    \item $4f + 1$ distinct hashes of valid blocks from the previous round, with the first hash being the author's own block
\end{itemize}

Each block is uniquely identified by the triplet (author, round, hash of the block contents). A block is considered valid if the signature is valid and the author is part of the validator set, there are at least $4f + 1$ hashes which point to distinct valid blocks from the previous round with the first being to the author's own previous block, and the block's causal history can be fully verified.

Honest validators maintain the integrity of their local DAG by accepting only valid blocks and rejecting malformed or invalid proposals. Additionally, validators protect causal consistency by incorporating a block only after recursively retrieving and verifying all blocks in the list of references. This ensures that the entirety of the DAG is legitimate by disallowing Byzantine validators to insert invalid blocks or fabricate references.

\subsection{Rounds and Waves}

Waves in \lmysticeti are two rounds long. The first round, called the propose round, is when validators create and disseminate their block proposals. The second round, called the decide round, sees validators evaluate the blocks from the propose round and make decisions based on them.

Wave pipelining is employed by starting a new wave with every round. This means that while validators are making decisions for blocks in one wave during a decide round, they are simultaneously proposing new blocks for the next wave. This overlapping execution produces maximum progress in the protocol, as there is always an active propose round generating new proposals while previous waves are being decided, with each propose round serving as the decide round for the preceding wave. Shoal \cite{shoal}, also using waves of two rounds, shares the exact same pipelined wave structure which can be seen in Figure~\ref{fig:rounds}.

While waves in Bullshark \cite{bullshark} and Shoal are also two rounds long, there is a major difference between them and \lmysticeti waves. Bullshark and Shoal both certify blocks. Thus, their two round waves equate to six message delays (three delays per round for certifications). Each \lmysticeti round is one message delay, so the waves are, in total, two message delays.

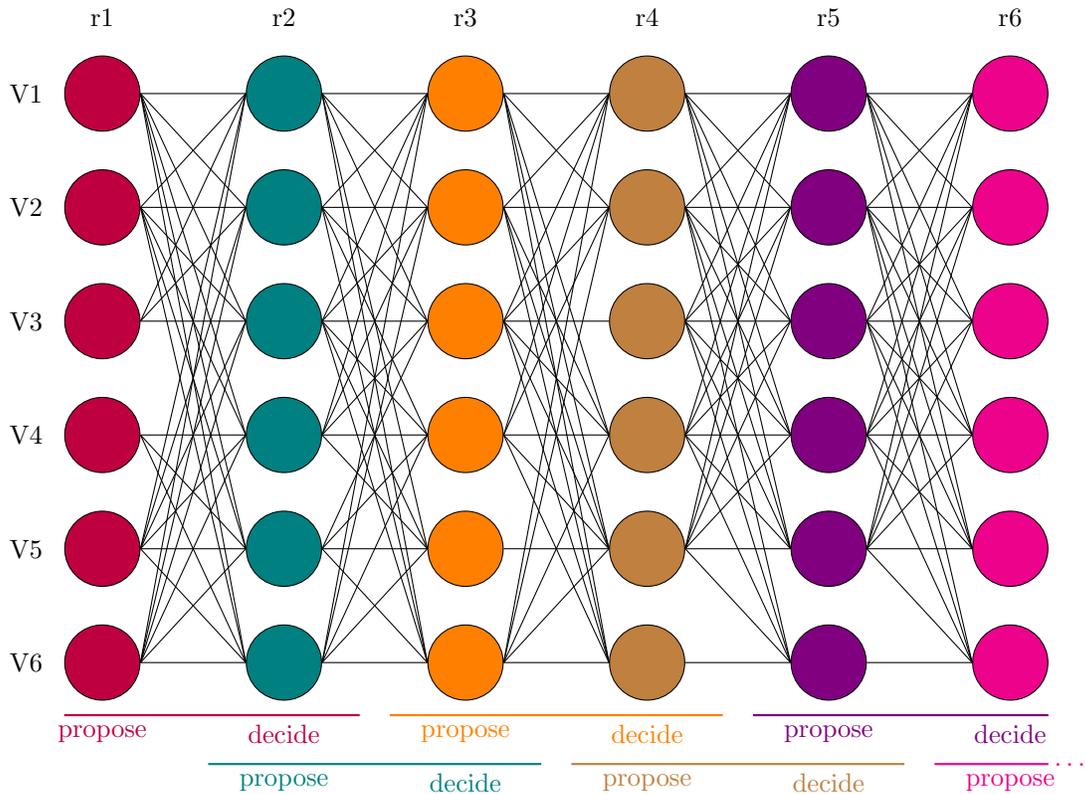
\begin{figure}[htbp]
\centering
\begin{tikzpicture}[scale=1, transform shape, node distance=0.5cm and 1.4cm,
    every node/.style={draw, minimum size=1cm, circle}]

\node[fill=purple] (n1) at (0, 0) {};
\node[fill=teal] (n2) [right=of n1] {};
\node[fill=orange] (n3) [right=of n2] {};
\node[fill=brown] (n4) [right=of n3] {};
\node[fill=purple] (n5) [below=of n1] {};
\node[fill=teal] (n6) [right=of n5] {};
\node[fill=orange] (n7) [right=of n6] {};
\node[fill=brown] (n8) [right=of n7] {};
\node[fill=purple] (n9) [below=of n5] {};
\node[fill=teal] (n10) [right=of n9] {};
\node[fill=orange] (n11) [right=of n10] {};
\node[fill=brown] (n12) [right=of n11] {};
\node[fill=purple] (n13) [below=of n9] {};
\node[fill=teal] (n14) [right=of n13] {};
\node[fill=orange] (n15) [right=of n14] {};
\node[fill=brown] (n16) [right=of n15] {};
\node[fill=purple] (n17) [below=of n13] {};
\node[fill=teal] (n18) [right=of n17] {};
\node[fill=orange] (n19) [right=of n18] {};
\node[fill=brown] (n20) [right=of n19] {};
\node[fill=purple] (n21) [below=of n17] {};
\node[fill=teal] (n22) [right=of n21] {};
\node[fill=orange] (n23) [right=of n22] {};
\node[fill=brown] (n24) [right=of n23] {};

\node[fill=violet] (n25) [right=of n4] {};
\node[fill=magenta] (n26) [right=of n25] {};
\node[fill=violet] (n27) [right=of n8] {};
\node[fill=magenta] (n28) [right=of n27] {};
\node[fill=violet] (n29) [right=of n12] {};
\node[fill=magenta] (n30) [right=of n29] {};
\node[fill=violet] (n31) [right=of n16] {};
\node[fill=magenta] (n32) [right=of n31] {};
\node[fill=violet] (n33) [right=of n20] {};
\node[fill=magenta] (n34) [right=of n33] {};
\node[fill=violet] (n35) [right=of n24] {};
\node[fill=magenta] (n36) [right=of n35] {};

\draw[-] (n4.east) -- (n25.west);
\draw[-] (n4.east) -- (n27.west);
\draw[-] (n4.east) -- (n29.west);
\draw[-] (n4.east) -- (n31.west);
\draw[-] (n4.east) -- (n33.west);
\draw[-] (n8.east) -- (n25.west);
\draw[-] (n8.east) -- (n27.west);
\draw[-] (n8.east) -- (n29.west);
\draw[-] (n8.east) -- (n31.west);
\draw[-] (n8.east) -- (n33.west);
\draw[-] (n12.east) -- (n25.west);
\draw[-] (n12.east) -- (n27.west);
\draw[-] (n12.east) -- (n29.west);
\draw[-] (n12.east) -- (n31.west);
\draw[-] (n12.east) -- (n33.west);
\draw[-] (n16.east) -- (n25.west);
\draw[-] (n16.east) -- (n27.west);
\draw[-] (n16.east) -- (n29.west);
\draw[-] (n16.east) -- (n31.west);
\draw[-] (n16.east) -- (n33.west);
\draw[-] (n20.east) -- (n25.west);
\draw[-] (n20.east) -- (n27.west);
\draw[-] (n20.east) -- (n29.west);
\draw[-] (n20.east) -- (n31.west);
\draw[-] (n20.east) -- (n33.west);
\draw[-] (n24.east) -- (n35.west);
\draw[-] (n8.east) -- (n35.west);
\draw[-] (n12.east) -- (n35.west);
\draw[-] (n16.east) -- (n35.west);
\draw[-] (n20.east) -- (n35.west);

\draw[-] (n25.east) -- (n26.west);
\draw[-] (n25.east) -- (n28.west);
\draw[-] (n25.east) -- (n30.west);
\draw[-] (n25.east) -- (n32.west);
\draw[-] (n25.east) -- (n34.west);
\draw[-] (n27.east) -- (n26.west);
\draw[-] (n27.east) -- (n28.west);
\draw[-] (n27.east) -- (n30.west);
\draw[-] (n27.east) -- (n32.west);
\draw[-] (n27.east) -- (n34.west);
\draw[-] (n29.east) -- (n26.west);
\draw[-] (n29.east) -- (n28.west);
\draw[-] (n29.east) -- (n30.west);
\draw[-] (n29.east) -- (n32.west);
\draw[-] (n29.east) -- (n34.west);
\draw[-] (n31.east) -- (n26.west);
\draw[-] (n31.east) -- (n28.west);
\draw[-] (n31.east) -- (n30.west);
\draw[-] (n31.east) -- (n32.west);
\draw[-] (n31.east) -- (n34.west);
\draw[-] (n33.east) -- (n26.west);
\draw[-] (n33.east) -- (n28.west);
\draw[-] (n33.east) -- (n30.west);
\draw[-] (n33.east) -- (n32.west);
\draw[-] (n33.east) -- (n34.west);
\draw[-] (n35.east) -- (n36.west);
\draw[-] (n27.east) -- (n36.west);
\draw[-] (n29.east) -- (n36.west);
\draw[-] (n31.east) -- (n36.west);
\draw[-] (n33.east) -- (n36.west);

\draw[-] (n3.east) -- (n4.west);
\draw[-] (n11.east) -- (n4.west);
\draw[-] (n11.east) -- (n12.west);
\draw[-] (n15.east) -- (n4.west);
\draw[-] (n23.east) -- (n4.west);
\draw[-] (n3.east) -- (n8.west);
\draw[-] (n11.east) -- (n8.west);
\draw[-] (n15.east) -- (n8.west);
\draw[-] (n23.east) -- (n8.west);
\draw[-] (n3.east) -- (n16.west);
\draw[-] (n11.east) -- (n16.west);
\draw[-] (n15.east) -- (n16.west);
\draw[-] (n23.east) -- (n16.west);
\draw[-] (n3.east) -- (n24.west);
\draw[-] (n11.east) -- (n24.west);
\draw[-] (n15.east) -- (n24.west);
\draw[-] (n23.east) -- (n24.west);
\draw[-] (n3.east) -- (n20.west);
\draw[-] (n11.east) -- (n20.west);
\draw[-] (n15.east) -- (n20.west);
\draw[-] (n19.east) -- (n20.west);
\draw[-] (n23.east) -- (n20.west);

\draw[-] (n18.east) -- (n11.west);
\draw[-] (n18.east) -- (n15.west);
\draw[-] (n18.east) -- (n19.west);
\draw[-] (n18.east) -- (n23.west);
\draw[-] (n2.east) -- (n3.west);
\draw[-] (n10.east) -- (n3.west);
\draw[-] (n14.east) -- (n3.west);
\draw[-] (n22.east) -- (n3.west);

\draw[-] (n2.east) -- (n11.west);
\draw[-] (n10.east) -- (n11.west);
\draw[-] (n14.east) -- (n11.west);
\draw[-] (n22.east) -- (n15.west);
\draw[-] (n10.east) -- (n15.west);
\draw[-] (n14.east) -- (n15.west);
\draw[-] (n2.east) -- (n19.west);
\draw[-] (n10.east) -- (n19.west);
\draw[-] (n14.east) -- (n19.west);
\draw[-] (n22.east) -- (n19.west);
\draw[-] (n2.east) -- (n23.west);
\draw[-] (n10.east) -- (n23.west);
\draw[-] (n22.east) -- (n23.west);

\draw[-] (n7.east) -- (n4.west);
\draw[-] (n7.east) -- (n8.west);
\draw[-] (n7.east) -- (n16.west);
\draw[-] (n7.east) -- (n20.west);
\draw[-] (n7.east) -- (n24.west);
\draw[-] (n6.east) -- (n3.west);
\draw[-] (n6.east) -- (n11.west);
\draw[-] (n6.east) -- (n15.west);
\draw[-] (n6.east) -- (n23.west);

\draw[-] (n13.east) -- (n18.west);
\draw[-] (n1.east) -- (n2.west);
\draw[-] (n5.east) -- (n2.west);
\draw[-] (n17.east) -- (n2.west);
\draw[-] (n21.east) -- (n2.west);
\draw[-] (n1.east) -- (n6.west);
\draw[-] (n5.east) -- (n6.west);
\draw[-] (n17.east) -- (n6.west);
\draw[-] (n21.east) -- (n6.west);
\draw[-] (n1.east) -- (n10.west);
\draw[-] (n5.east) -- (n10.west);
\draw[-] (n17.east) -- (n10.west);
\draw[-] (n21.east) -- (n10.west);
\draw[-] (n1.east) -- (n14.west);
\draw[-] (n5.east) -- (n14.west);
\draw[-] (n17.east) -- (n14.west);
\draw[-] (n21.east) -- (n14.west);
\draw[-] (n1.east) -- (n18.west);
\draw[-] (n5.east) -- (n18.west);
\draw[-] (n17.east) -- (n18.west);
\draw[-] (n21.east) -- (n18.west);
\draw[-] (n1.east) -- (n22.west);
\draw[-] (n5.east) -- (n22.west);
\draw[-] (n17.east) -- (n22.west);
\draw[-] (n21.east) -- (n22.west);

\draw[-] (n13.east) -- (n14.west);
\draw[-] (n13.east) -- (n22.west);
\draw[-] (n14.east) -- (n7.west);
\draw[-] (n22.east) -- (n7.west);

\draw[-] (n9.east) -- (n2.west);
\draw[-] (n9.east) -- (n6.west);
\draw[-] (n9.east) -- (n10.west);
\draw[-] (n2.east) -- (n7.west);
\draw[-] (n10.east) -- (n7.west);
\draw[-] (n6.east) -- (n7.west);

\node[draw=none, anchor=east] at (n1.west) {V1};
\node[draw=none, anchor=east] at (n5.west) {V2};
\node[draw=none, anchor=east] at (n9.west) {V3};
\node[draw=none, anchor=east] at (n13.west) {V4};
\node[draw=none, anchor=east] at (n17.west) {V5};
\node[draw=none, anchor=east] at (n21.west) {V6};

\node[draw=none, anchor=south] at (n1.north) {r1};
\node[draw=none, anchor=south] at (n2.north) {r2};
\node[draw=none, anchor=south] at (n3.north) {r3};
\node[draw=none, anchor=south] at (n4.north) {r4};
\node[draw=none, anchor=south] at (n25.north) {r5};
\node[draw=none, anchor=south] at (n26.north) {r6};

\draw[thick, -, purple] ([yshift=-0.7cm]n21.west) -- ([yshift=-0.7cm, xshift=0.5cm]n22.east);
\node[draw=none, purple] at ([yshift=-0.95cm]n21) {propose};
\node[draw=none, purple] at ([yshift=-0.95cm]n22) {decide};

\draw[thick, -, orange] ([yshift=-0.7cm, xshift=-0.5cm]n23.west) -- ([yshift=-0.7cm, xshift=0.5cm]n24.east);
\node[draw=none, orange] at ([yshift=-0.95cm]n23) {propose};
\node[draw=none, orange] at ([yshift=-0.95cm]n24) {decide};

\draw[thick, -, violet] ([yshift=-0.7cm, xshift=-0.5cm]n35.west) -- ([yshift=-0.7cm]n36.east);
\node[draw=none, violet] at ([yshift=-0.95cm]n35) {propose};
\node[draw=none, violet] at ([yshift=-0.95cm]n36) {decide};

\draw[thick, -, teal] ([yshift=-1.35cm, xshift=-0.5cm]n22.west) -- ([yshift=-1.35cm, xshift=0.5cm]n23.east);
\node[draw=none, teal] at ([yshift=-1.6cm]n22) {propose};
\node[draw=none, teal] at ([yshift=-1.6cm]n23) {decide};

\draw[thick, -, brown] ([yshift=-1.35cm, xshift=-0.5cm]n24.west) -- ([yshift=-1.35cm, xshift=0.5cm]n35.east);
\node[draw=none, brown] at ([yshift=-1.6cm]n24) {propose};
\node[draw=none, brown] at ([yshift=-1.6cm]n35) {decide};

\draw[thick, -, magenta] ([yshift=-1.35cm, xshift=-0.5cm]n36.west) -- ([yshift=-1.35cm]n36.east) node[draw=none, pos=1, right, magenta, xshift=-0.2cm] {\dots};
\node[draw=none, magenta] at ([yshift=-1.6cm]n36) {propose};

\end{tikzpicture}
\vspace{-0.5cm}
\caption{\textbf{DAG Structure}. The DAG shows six validators (V1-V6) across six rounds (r1-r6). The colored lines at the bottom indicate waves, where each wave consists of a propose round followed by a decide round. Blocks are colored based on the propose round they were created during. Multiple waves execute in parallel through pipelining, with a new wave beginning each round.}
\label{fig:rounds}
\end{figure}

\chapter{Protocol}\label{chap:protocol}

\section{Leader Slots}

Leader slots in \lmysticeti are represented by a pair (validator, round) and can either be empty or contain the validator's proposal for the respective round. If the chosen validator is Byzantine, they may have equivocated, in which case the slot contains more than one block.

Multiple leader slots can exist per round, enabling parallel leader block proposals. Each slot is one of three states: \textbf{committed}, \textbf{skipped}, or \textbf{undecided}. All slots begin in the undecided state, and the protocol's objective is to classify them as either committed or skipped. The commit state indicates that the slot's block should be included in the total ordering, while the skip state allows the protocol to exclude slots from crashed or Byzantine validators. Crucially, the undecided state forces subsequent leader slots to wait, preventing unsafe commitments. This is discussed in more detail later. The set up mirrors related work \cite{mysticeti,mahimahi}.

There is a ranking amongst all leader slots as within each round, leaders are ranked. Validators wait for the leader proposals for up to $2 \cdot \Delta$ before timing out and proceeding to the next round. The assignment of validators to slots follows a round robin schedule, presented in Algorithm~\ref{alg:helper}, ensuring that every validator will be assigned to the highest-ranked leader slot exactly once during a set of $5f + 1$ consecutive rounds. This rotation favors a fair, uniform leadership distribution while maintaining the predictable ordering necessary for consensus. Both the timeout parameter and rotation schedule contribute to the formal liveness guarantee in Section~\ref{sec:liveness}.

The number of leader slots per round is configurable through the protocol given in Algorithm~\ref{alg:lmys}, allowing the system to be altered based on the expected network behavior and desired throughput characteristics. Increasing the number of leaders in a round may improve throughput as more blocks have the possibility of being committed in a round. However, it could also slow down consensus as validators need to wait for all leaders before producing blocks in the next round.

\section{Decision Rule}

Before detailing the decision rule, some terminology must be established. A block \textbf{supports} a leader block if it includes the leader block as a parent (a reference). A block \textbf{blames} a leader block if it does not include the leader block as a parent. Note that parents can only be from the previous round, unlike in earlier consensus protocols \cite{dag-rider,hashgraph} where a block's parents could be from any prior round.

An \textbf{anchor}, with relation to a leader block, is the first committed or undecided leader block outside of the current leader block's wave. Formally, let $B$ be a leader block in round $r$. The search for $B$'s anchor begins in round $r + 2$ as $B$'s wave ends with round $r + 1$. Leader blocks are checked in rank order from highest to lowest until one is either undecided or committed. If the leaders in a round are exhausted, then the search continues with the next higher round. Anchors may not exist. For example, if round $r + 1$ is the highest round in the DAG, then no anchors will be found for leaders in round $r$ as the wave ends the DAG. 

\lmysticeti's decision rule contains three steps and is illustrated in full in Figure~\ref{fig:decision}.

\subsection{Step 1: Direct Decision}

The direct decision rule is acted upon right after a wave ends. Let $B$ be a leader block in round $r$. $B$ is committed if it has $4f + 1$ supports or $B$ is skipped if it has $4f + 1$ blames. In Figure~\ref{fig:decision}, L3a and L2b are committed and L3b is skipped by the direct rule. If $B$ fails to be decided by the direct decision rule, the indirect decision rule commences.

Note that crash faults are handled effortlessly within this step. If a leader crashes and fails to propose a block, the corresponding leader slot remains empty. Since validators can only support blocks that exist, an empty slot will naturally accumulate $4f + 1$ blames, causing it to be skipped. No special crash detection is required.

\subsection{Step 2: Indirect Decision}

The indirect decision rule starts by locating an anchor for $B$. If the anchor is undecided or cannot be found, then $B$ stays undecided. Otherwise, the causal history of the anchor is investigated. If there are $2f + 1$ supports for $B$ in the history of the anchor, then $B$ is committed. If there are not enough supports, then $B$ is skipped. In Figure~\ref{fig:decision}, L1b is committed, L1a is skipped, and L2a is kept undecided by the indirect decision rule. By leveraging the predefined ranking of leader slots, any indirect decision is carried forward through its anchor.

\subsection{Step 3: Sequencing}

Steps 1 and 2 start at the lowest ranked leader block in the highest round and proceed to the highest ranked leader block in that round before progressing to the lowest ranked leader block in the previous round. This continues until a leader block is encountered which is already in the total ordering. In Figure~\ref{fig:decision}, assuming that no leader blocks are currently in the total ordering, the order of leader blocks processed is L4b, L3b, L3a, L2b, L2a, L1b, L1a. 

Validators must then sequence the processed leader slots. They begin by ordering leaders by round then rank, the reverse of the processing order. Skipped leader blocks are not included and the sequence terminates at the first undecided leader block. In Figure~\ref{fig:decision}, after the decision rules are applied, the ordering proceeds as follows: L1a is skipped and excluded, L1b is committed and included, and the undecided L2a terminates the series. Thus, only L1b is in the leader sequence.

The unique sub-DAG of each (committed) leader block in the leader sequence is then linearized.

\captionsetup[subfigure]{justification=centering, skip=-0.75cm}
\begin{figure}[htbp]
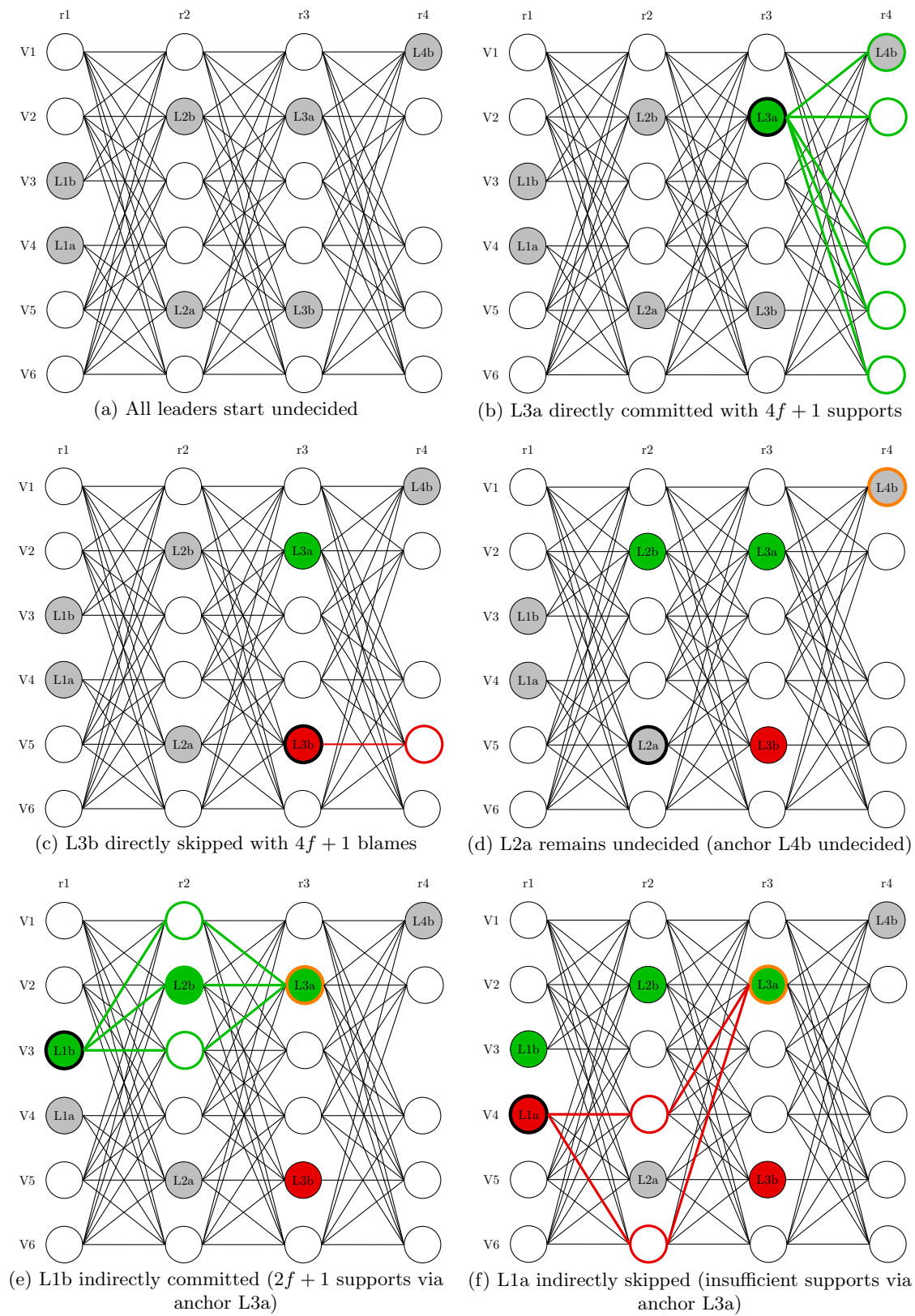

\centering

\begin{subfigure}[t]{0.49\textwidth}
  \centering
  \begin{tikzpicture}[scale=0.6, transform shape, node distance=0.75cm and 2.25cm,
    every node/.style={draw, minimum size=1cm}, circle]

\node[] (n1) at (0, 0) {};
\node[] (n2) [right=of n1] {};
\node[] (n3) [right=of n2] {};
\node[fill=gray!50] (n4) [right=of n3] {L4b};
\node[] (n5) [below=of n1] {};
\node[fill=gray!50] (n6) [right=of n5] {L2b};
\node[fill=gray!50] (n7) [right=of n6] {L3a};
\node[] (n8) [right=of n7] {};
\node[fill=gray!50] (n9) [below=of n5] {L1b};
\node[] (n10) [right=of n9] {};
\node[] (n11) [right=of n10] {};
\node[draw=none] (n12) [right=of n11] {};
\node[fill=gray!50] (n13) [below=of n9] {L1a};
\node[] (n14) [right=of n13] {};
\node[] (n15) [right=of n14] {};
\node[] (n16) [right=of n15] {};
\node[] (n17) [below=of n13] {};
\node[fill=gray!50] (n18) [right=of n17] {L2a};
\node[fill=gray!50] (n19) [right=of n18] {L3b};
\node[] (n20) [right=of n19] {};
\node[] (n21) [below=of n17] {};
\node[] (n22) [right=of n21] {};
\node[] (n23) [right=of n22] {};
\node[] (n24) [right=of n23] {};

\draw[-] (n3.east) -- (n4.west);
\draw[-] (n11.east) -- (n4.west);
\draw[-] (n15.east) -- (n4.west);
\draw[-] (n23.east) -- (n4.west);
\draw[-] (n3.east) -- (n8.west);
\draw[-] (n11.east) -- (n8.west);
\draw[-] (n15.east) -- (n8.west);
\draw[-] (n23.east) -- (n8.west);
\draw[-] (n3.east) -- (n16.west);
\draw[-] (n11.east) -- (n16.west);
\draw[-] (n15.east) -- (n16.west);
\draw[-] (n23.east) -- (n16.west);
\draw[-] (n3.east) -- (n24.west);
\draw[-] (n11.east) -- (n24.west);
\draw[-] (n15.east) -- (n24.west);
\draw[-] (n23.east) -- (n24.west);
\draw[-] (n3.east) -- (n20.west);
\draw[-] (n11.east) -- (n20.west);
\draw[-] (n15.east) -- (n20.west);
\draw[-] (n19.east) -- (n20.west);
\draw[-] (n23.east) -- (n20.west);

\draw[-] (n18.east) -- (n11.west);
\draw[-] (n18.east) -- (n15.west);
\draw[-] (n18.east) -- (n19.west);
\draw[-] (n18.east) -- (n23.west);
\draw[-] (n2.east) -- (n3.west);
\draw[-] (n10.east) -- (n3.west);
\draw[-] (n14.east) -- (n3.west);
\draw[-] (n22.east) -- (n3.west);

\draw[-] (n2.east) -- (n11.west);
\draw[-] (n10.east) -- (n11.west);
\draw[-] (n14.east) -- (n11.west);
\draw[-] (n22.east) -- (n15.west);
\draw[-] (n10.east) -- (n15.west);
\draw[-] (n14.east) -- (n15.west);
\draw[-] (n2.east) -- (n19.west);
\draw[-] (n10.east) -- (n19.west);
\draw[-] (n14.east) -- (n19.west);
\draw[-] (n22.east) -- (n19.west);
\draw[-] (n2.east) -- (n23.west);
\draw[-] (n10.east) -- (n23.west);
\draw[-] (n22.east) -- (n23.west);

\draw[-] (n7.east) -- (n4.west);
\draw[-] (n7.east) -- (n8.west);
\draw[-] (n7.east) -- (n16.west);
\draw[-] (n7.east) -- (n20.west);
\draw[-] (n7.east) -- (n24.west);
\draw[-] (n6.east) -- (n3.west);
\draw[-] (n6.east) -- (n11.west);
\draw[-] (n6.east) -- (n15.west);
\draw[-] (n6.east) -- (n23.west);

\draw[-] (n13.east) -- (n18.west);
\draw[-] (n1.east) -- (n2.west);
\draw[-] (n5.east) -- (n2.west);
\draw[-] (n17.east) -- (n2.west);
\draw[-] (n21.east) -- (n2.west);
\draw[-] (n1.east) -- (n6.west);
\draw[-] (n5.east) -- (n6.west);
\draw[-] (n17.east) -- (n6.west);
\draw[-] (n21.east) -- (n6.west);
\draw[-] (n1.east) -- (n10.west);
\draw[-] (n5.east) -- (n10.west);
\draw[-] (n17.east) -- (n10.west);
\draw[-] (n21.east) -- (n10.west);
\draw[-] (n1.east) -- (n14.west);
\draw[-] (n5.east) -- (n14.west);
\draw[-] (n17.east) -- (n14.west);
\draw[-] (n21.east) -- (n14.west);
\draw[-] (n1.east) -- (n18.west);
\draw[-] (n5.east) -- (n18.west);
\draw[-] (n17.east) -- (n18.west);
\draw[-] (n21.east) -- (n18.west);
\draw[-] (n1.east) -- (n22.west);
\draw[-] (n5.east) -- (n22.west);
\draw[-] (n17.east) -- (n22.west);
\draw[-] (n21.east) -- (n22.west);

\draw[-] (n13.east) -- (n14.west);
\draw[-] (n13.east) -- (n22.west);
\draw[-] (n14.east) -- (n7.west);
\draw[-] (n22.east) -- (n7.west);


\draw[-] (n9.east) -- (n2.west);
\draw[-] (n9.east) -- (n6.west);
\draw[-] (n9.east) -- (n10.west);
\draw[-] (n2.east) -- (n7.west);
\draw[-] (n10.east) -- (n7.west);
\draw[-] (n6.east) -- (n7.west);


\input{assets/decide/setup}

\end{tikzpicture}
  \caption{All leaders start undecided}
  \label{fig:decide1}
\end{subfigure}
\hfill
\begin{subfigure}[t]{0.49\textwidth}
  \centering
  \begin{tikzpicture}[scale=0.6, transform shape, node distance=0.75cm and 2.25cm,
    every node/.style={draw, minimum size=1cm}, circle]

\node[] (n1) at (0, 0) {};
\node[] (n2) [right=of n1] {};
\node[] (n3) [right=of n2] {};
\node[fill=gray!50, draw=green!75!black, very thick] (n4) [right=of n3] {L4b};
\node[] (n5) [below=of n1] {};
\node[fill=gray!50] (n6) [right=of n5] {L2b};
\node[fill=green!75!black, draw=black, ultra thick] (n7) [right=of n6] {L3a};
\node[draw=green!75!black, very thick] (n8) [right=of n7] {};
\node[fill=gray!50] (n9) [below=of n5] {L1b};
\node[] (n10) [right=of n9] {};
\node[] (n11) [right=of n10] {};
\node[draw=none] (n12) [right=of n11] {};
\node[fill=gray!50] (n13) [below=of n9] {L1a};
\node[] (n14) [right=of n13] {};
\node[] (n15) [right=of n14] {};
\node[draw=green!75!black, very thick] (n16) [right=of n15] {};
\node[] (n17) [below=of n13] {};
\node[fill=gray!50] (n18) [right=of n17] {L2a};
\node[fill=gray!50] (n19) [right=of n18] {L3b};
\node[draw=green!75!black, very thick] (n20) [right=of n19] {};
\node[] (n21) [below=of n17] {};
\node[] (n22) [right=of n21] {};
\node[] (n23) [right=of n22] {};
\node[draw=green!75!black, very thick] (n24) [right=of n23] {};

\draw[-] (n3.east) -- (n4.west);
\draw[-] (n11.east) -- (n4.west);
\draw[-] (n15.east) -- (n4.west);
\draw[-] (n23.east) -- (n4.west);
\draw[-] (n3.east) -- (n8.west);
\draw[-] (n11.east) -- (n8.west);
\draw[-] (n15.east) -- (n8.west);
\draw[-] (n23.east) -- (n8.west);
\draw[-] (n3.east) -- (n16.west);
\draw[-] (n11.east) -- (n16.west);
\draw[-] (n15.east) -- (n16.west);
\draw[-] (n23.east) -- (n16.west);
\draw[-] (n3.east) -- (n24.west);
\draw[-] (n11.east) -- (n24.west);
\draw[-] (n15.east) -- (n24.west);
\draw[-] (n23.east) -- (n24.west);
\draw[-] (n3.east) -- (n20.west);
\draw[-] (n11.east) -- (n20.west);
\draw[-] (n15.east) -- (n20.west);
\draw[-] (n19.east) -- (n20.west);
\draw[-] (n23.east) -- (n20.west);

\draw[-] (n18.east) -- (n11.west);
\draw[-] (n18.east) -- (n15.west);
\draw[-] (n18.east) -- (n19.west);
\draw[-] (n18.east) -- (n23.west);
\draw[-] (n2.east) -- (n3.west);
\draw[-] (n10.east) -- (n3.west);
\draw[-] (n14.east) -- (n3.west);
\draw[-] (n22.east) -- (n3.west);

\draw[-] (n2.east) -- (n11.west);
\draw[-] (n10.east) -- (n11.west);
\draw[-] (n14.east) -- (n11.west);
\draw[-] (n22.east) -- (n15.west);
\draw[-] (n10.east) -- (n15.west);
\draw[-] (n14.east) -- (n15.west);
\draw[-] (n2.east) -- (n19.west);
\draw[-] (n10.east) -- (n19.west);
\draw[-] (n14.east) -- (n19.west);
\draw[-] (n22.east) -- (n19.west);
\draw[-] (n2.east) -- (n23.west);
\draw[-] (n10.east) -- (n23.west);
\draw[-] (n22.east) -- (n23.west);

\draw[-] (n13.east) -- (n18.west);
\draw[-] (n1.east) -- (n2.west);
\draw[-] (n5.east) -- (n2.west);
\draw[-] (n17.east) -- (n2.west);
\draw[-] (n21.east) -- (n2.west);
\draw[-] (n1.east) -- (n6.west);
\draw[-] (n5.east) -- (n6.west);
\draw[-] (n17.east) -- (n6.west);
\draw[-] (n21.east) -- (n6.west);
\draw[-] (n1.east) -- (n10.west);
\draw[-] (n5.east) -- (n10.west);
\draw[-] (n17.east) -- (n10.west);
\draw[-] (n21.east) -- (n10.west);
\draw[-] (n1.east) -- (n14.west);
\draw[-] (n5.east) -- (n14.west);
\draw[-] (n17.east) -- (n14.west);
\draw[-] (n21.east) -- (n14.west);
\draw[-] (n1.east) -- (n18.west);
\draw[-] (n5.east) -- (n18.west);
\draw[-] (n17.east) -- (n18.west);
\draw[-] (n21.east) -- (n18.west);
\draw[-] (n1.east) -- (n22.west);
\draw[-] (n5.east) -- (n22.west);
\draw[-] (n17.east) -- (n22.west);
\draw[-] (n21.east) -- (n22.west);

\draw[-] (n13.east) -- (n14.west);
\draw[-] (n13.east) -- (n22.west);
\draw[-] (n14.east) -- (n7.west);
\draw[-] (n22.east) -- (n7.west);

\draw[-] (n9.east) -- (n2.west);
\draw[-] (n9.east) -- (n6.west);
\draw[-] (n9.east) -- (n10.west);
\draw[-] (n2.east) -- (n7.west);
\draw[-] (n10.east) -- (n7.west);


\draw[-, green!75!black, very thick] (n7.east) -- (n4.west);
\draw[-, green!75!black, very thick] (n7.east) -- (n8.west);
\draw[-, green!75!black, very thick] (n7.east) -- (n16.west);
\draw[-, green!75!black, very thick] (n7.east) -- (n20.west);
\draw[-, green!75!black, very thick] (n7.east) -- (n24.west);
\draw[-] (n6.east) -- (n3.west);
\draw[-] (n6.east) -- (n11.west);
\draw[-] (n6.east) -- (n15.west);
\draw[-] (n6.east) -- (n23.west);
\draw[-] (n6.east) -- (n7.west);

\input{assets/decide/setup}

\end{tikzpicture}
  \caption{L3a directly committed with $4f+1$ supports}
  \label{fig:decide2}
\end{subfigure}

\vspace{0.05cm}

\begin{subfigure}[t]{0.49\textwidth}
  \centering
  \begin{tikzpicture}[scale=0.6, transform shape, node distance=0.75cm and 2.25cm,
    every node/.style={draw, minimum size=1cm}, circle]

\node[] (n1) at (0, 0) {};
\node[] (n2) [right=of n1] {};
\node[] (n3) [right=of n2] {};
\node[fill=gray!50] (n4) [right=of n3] {L4b};
\node[] (n5) [below=of n1] {};
\node[fill=gray!50] (n6) [right=of n5] {L2b};
\node[fill=green!75!black] (n7) [right=of n6] {L3a};
\node[] (n8) [right=of n7] {};
\node[fill=gray!50] (n9) [below=of n5] {L1b};
\node[] (n10) [right=of n9] {};
\node[] (n11) [right=of n10] {};
\node[draw=none] (n12) [right=of n11] {};
\node[fill=gray!50] (n13) [below=of n9] {L1a};
\node[] (n14) [right=of n13] {};
\node[] (n15) [right=of n14] {};
\node[] (n16) [right=of n15] {};
\node[] (n17) [below=of n13] {};
\node[fill=gray!50] (n18) [right=of n17] {L2a};
\node[fill=red!90!black, draw=black, ultra thick] (n19) [right=of n18] {L3b};
\node[draw=red!90!black, very thick] (n20) [right=of n19] {};
\node[] (n21) [below=of n17] {};
\node[] (n22) [right=of n21] {};
\node[] (n23) [right=of n22] {};
\node[] (n24) [right=of n23] {};

\draw[-] (n3.east) -- (n4.west);
\draw[-] (n11.east) -- (n4.west);
\draw[-] (n15.east) -- (n4.west);
\draw[-] (n23.east) -- (n4.west);
\draw[-] (n3.east) -- (n8.west);
\draw[-] (n11.east) -- (n8.west);
\draw[-] (n15.east) -- (n8.west);
\draw[-] (n23.east) -- (n8.west);
\draw[-] (n3.east) -- (n16.west);
\draw[-] (n11.east) -- (n16.west);
\draw[-] (n15.east) -- (n16.west);
\draw[-] (n23.east) -- (n16.west);
\draw[-] (n3.east) -- (n24.west);
\draw[-] (n11.east) -- (n24.west);
\draw[-] (n15.east) -- (n24.west);
\draw[-] (n23.east) -- (n24.west);
\draw[-] (n3.east) -- (n20.west);
\draw[-] (n11.east) -- (n20.west);
\draw[-] (n15.east) -- (n20.west);
\draw[-] (n23.east) -- (n20.west);

\draw[-] (n18.east) -- (n11.west);
\draw[-] (n18.east) -- (n15.west);
\draw[-] (n18.east) -- (n19.west);
\draw[-] (n18.east) -- (n23.west);
\draw[-] (n2.east) -- (n3.west);
\draw[-] (n10.east) -- (n3.west);
\draw[-] (n14.east) -- (n3.west);
\draw[-] (n22.east) -- (n3.west);

\draw[-] (n2.east) -- (n11.west);
\draw[-] (n10.east) -- (n11.west);
\draw[-] (n14.east) -- (n11.west);
\draw[-] (n22.east) -- (n15.west);
\draw[-] (n10.east) -- (n15.west);
\draw[-] (n14.east) -- (n15.west);
\draw[-] (n2.east) -- (n19.west);
\draw[-] (n10.east) -- (n19.west);
\draw[-] (n14.east) -- (n19.west);
\draw[-] (n22.east) -- (n19.west);
\draw[-] (n2.east) -- (n23.west);
\draw[-] (n10.east) -- (n23.west);
\draw[-] (n22.east) -- (n23.west);

\draw[-] (n13.east) -- (n18.west);
\draw[-] (n1.east) -- (n2.west);
\draw[-] (n5.east) -- (n2.west);
\draw[-] (n17.east) -- (n2.west);
\draw[-] (n21.east) -- (n2.west);
\draw[-] (n1.east) -- (n6.west);
\draw[-] (n5.east) -- (n6.west);
\draw[-] (n17.east) -- (n6.west);
\draw[-] (n21.east) -- (n6.west);
\draw[-] (n1.east) -- (n10.west);
\draw[-] (n5.east) -- (n10.west);
\draw[-] (n17.east) -- (n10.west);
\draw[-] (n21.east) -- (n10.west);
\draw[-] (n1.east) -- (n14.west);
\draw[-] (n5.east) -- (n14.west);
\draw[-] (n17.east) -- (n14.west);
\draw[-] (n21.east) -- (n14.west);
\draw[-] (n1.east) -- (n18.west);
\draw[-] (n5.east) -- (n18.west);
\draw[-] (n17.east) -- (n18.west);
\draw[-] (n21.east) -- (n18.west);
\draw[-] (n1.east) -- (n22.west);
\draw[-] (n5.east) -- (n22.west);
\draw[-] (n17.east) -- (n22.west);
\draw[-] (n21.east) -- (n22.west);

\draw[-] (n13.east) -- (n14.west);
\draw[-] (n13.east) -- (n22.west);
\draw[-] (n14.east) -- (n7.west);
\draw[-] (n22.east) -- (n7.west);

\draw[-] (n9.east) -- (n2.west);
\draw[-] (n9.east) -- (n6.west);
\draw[-] (n9.east) -- (n10.west);
\draw[-] (n2.east) -- (n7.west);
\draw[-] (n10.east) -- (n7.west);


\draw[-] (n7.east) -- (n4.west);
\draw[-] (n7.east) -- (n8.west);
\draw[-] (n7.east) -- (n16.west);
\draw[-] (n7.east) -- (n20.west);
\draw[-] (n7.east) -- (n24.west);
\draw[-] (n6.east) -- (n3.west);
\draw[-] (n6.east) -- (n11.west);
\draw[-] (n6.east) -- (n15.west);
\draw[-] (n6.east) -- (n23.west);
\draw[-] (n6.east) -- (n7.west);

\draw[-, red!90!black, thick] (n19.east) -- (n20.west);

\input{assets/decide/setup}

\end{tikzpicture}
  \caption{L3b directly skipped with $4f+1$ blames}
  \label{fig:decide3}
\end{subfigure}
\hfill
\begin{subfigure}[t]{0.49\textwidth}
  \centering
  \begin{tikzpicture}[scale=0.6, transform shape, node distance=0.75cm and 2.25cm,
    every node/.style={draw, minimum size=1cm}, circle]

\node[] (n1) at (0, 0) {};
\node[] (n2) [right=of n1] {};
\node[] (n3) [right=of n2] {};
\node[fill=gray!50, draw=orange, ultra thick] (n4) [right=of n3] {L4b};
\node[] (n5) [below=of n1] {};
\node[fill=green!75!black] (n6) [right=of n5] {L2b};
\node[fill=green!75!black] (n7) [right=of n6] {L3a};
\node[] (n8) [right=of n7] {};
\node[fill=gray!50] (n9) [below=of n5] {L1b};
\node[] (n10) [right=of n9] {};
\node[] (n11) [right=of n10] {};
\node[draw=none] (n12) [right=of n11] {};
\node[fill=gray!50] (n13) [below=of n9] {L1a};
\node[] (n14) [right=of n13] {};
\node[] (n15) [right=of n14] {};
\node[] (n16) [right=of n15] {};
\node[] (n17) [below=of n13] {};
\node[fill=gray!50, draw=black, ultra thick] (n18) [right=of n17] {L2a};
\node[fill=red!90!black] (n19) [right=of n18] {L3b};
\node[] (n20) [right=of n19] {};
\node[] (n21) [below=of n17] {};
\node[] (n22) [right=of n21] {};
\node[] (n23) [right=of n22] {};
\node[] (n24) [right=of n23] {};

\draw[-] (n3.east) -- (n4.west);
\draw[-] (n11.east) -- (n4.west);
\draw[-] (n15.east) -- (n4.west);
\draw[-] (n23.east) -- (n4.west);
\draw[-] (n3.east) -- (n8.west);
\draw[-] (n11.east) -- (n8.west);
\draw[-] (n15.east) -- (n8.west);
\draw[-] (n23.east) -- (n8.west);
\draw[-] (n3.east) -- (n16.west);
\draw[-] (n11.east) -- (n16.west);
\draw[-] (n15.east) -- (n16.west);
\draw[-] (n23.east) -- (n16.west);
\draw[-] (n3.east) -- (n24.west);
\draw[-] (n11.east) -- (n24.west);
\draw[-] (n15.east) -- (n24.west);
\draw[-] (n23.east) -- (n24.west);
\draw[-] (n3.east) -- (n20.west);
\draw[-] (n11.east) -- (n20.west);
\draw[-] (n15.east) -- (n20.west);
\draw[-] (n19.east) -- (n20.west);
\draw[-] (n23.east) -- (n20.west);

\draw[-] (n18.east) -- (n11.west);
\draw[-] (n18.east) -- (n15.west);
\draw[-] (n18.east) -- (n19.west);
\draw[-] (n18.east) -- (n23.west);
\draw[-] (n2.east) -- (n3.west);
\draw[-] (n10.east) -- (n3.west);
\draw[-] (n14.east) -- (n3.west);
\draw[-] (n22.east) -- (n3.west);

\draw[-] (n2.east) -- (n11.west);
\draw[-] (n10.east) -- (n11.west);
\draw[-] (n14.east) -- (n11.west);
\draw[-] (n22.east) -- (n15.west);
\draw[-] (n10.east) -- (n15.west);
\draw[-] (n14.east) -- (n15.west);
\draw[-] (n2.east) -- (n19.west);
\draw[-] (n10.east) -- (n19.west);
\draw[-] (n14.east) -- (n19.west);
\draw[-] (n22.east) -- (n19.west);
\draw[-] (n2.east) -- (n23.west);
\draw[-] (n10.east) -- (n23.west);
\draw[-] (n22.east) -- (n23.west);

\draw[-] (n6.east) -- (n3.west);
\draw[-] (n6.east) -- (n11.west);
\draw[-] (n6.east) -- (n15.west);
\draw[-] (n6.east) -- (n23.west);

\draw[-] (n13.east) -- (n18.west);
\draw[-] (n1.east) -- (n2.west);
\draw[-] (n5.east) -- (n2.west);
\draw[-] (n17.east) -- (n2.west);
\draw[-] (n21.east) -- (n2.west);
\draw[-] (n1.east) -- (n6.west);
\draw[-] (n5.east) -- (n6.west);
\draw[-] (n17.east) -- (n6.west);
\draw[-] (n21.east) -- (n6.west);
\draw[-] (n1.east) -- (n10.west);
\draw[-] (n5.east) -- (n10.west);
\draw[-] (n17.east) -- (n10.west);
\draw[-] (n21.east) -- (n10.west);
\draw[-] (n1.east) -- (n14.west);
\draw[-] (n5.east) -- (n14.west);
\draw[-] (n17.east) -- (n14.west);
\draw[-] (n21.east) -- (n14.west);
\draw[-] (n1.east) -- (n18.west);
\draw[-] (n5.east) -- (n18.west);
\draw[-] (n17.east) -- (n18.west);
\draw[-] (n21.east) -- (n18.west);
\draw[-] (n1.east) -- (n22.west);
\draw[-] (n5.east) -- (n22.west);
\draw[-] (n17.east) -- (n22.west);
\draw[-] (n21.east) -- (n22.west);

\draw[-] (n13.east) -- (n14.west);
\draw[-] (n13.east) -- (n22.west);
\draw[-] (n14.east) -- (n7.west);
\draw[-] (n22.east) -- (n7.west);

\draw[-] (n9.east) -- (n2.west);
\draw[-] (n9.east) -- (n6.west);
\draw[-] (n9.east) -- (n10.west);
\draw[-] (n2.east) -- (n7.west);
\draw[-] (n10.east) -- (n7.west);
\draw[-] (n6.east) -- (n7.west);

\draw[-] (n7.east) -- (n4.west);
\draw[-] (n7.east) -- (n8.west);
\draw[-] (n7.east) -- (n16.west);
\draw[-] (n7.east) -- (n20.west);
\draw[-] (n7.east) -- (n24.west);

\input{assets/decide/setup}

\end{tikzpicture}
  \caption{L2a remains undecided (anchor L4b undecided)}
  \label{fig:decide4}
\end{subfigure}

\vspace{0.05cm}

\begin{subfigure}[t]{0.49\textwidth}
  \centering
  \begin{tikzpicture}[scale=0.6, transform shape, node distance=0.75cm and 2.25cm,
    every node/.style={draw, minimum size=1cm}, circle]

\node[] (n1) at (0, 0) {};
\node[draw=green!75!black, very thick] (n2) [right=of n1] {};
\node[] (n3) [right=of n2] {};
\node[fill=gray!50] (n4) [right=of n3] {L4b};
\node[] (n5) [below=of n1] {};
\node[fill=green!75!black, draw=green!75!black, very thick] (n6) [right=of n5] {L2b};
\node[fill=green!75!black, draw=orange, ultra thick] (n7) [right=of n6] {L3a};
\node[] (n8) [right=of n7] {};
\node[fill=green!75!black, draw=black, ultra thick] (n9) [below=of n5] {L1b};
\node[draw=green!75!black, very thick] (n10) [right=of n9] {};
\node[] (n11) [right=of n10] {};
\node[draw=none] (n12) [right=of n11] {};
\node[fill=gray!50] (n13) [below=of n9] {L1a};
\node[] (n14) [right=of n13] {};
\node[] (n15) [right=of n14] {};
\node[] (n16) [right=of n15] {};
\node[] (n17) [below=of n13] {};
\node[fill=gray!50] (n18) [right=of n17] {L2a};
\node[fill=red!90!black] (n19) [right=of n18] {L3b};
\node[] (n20) [right=of n19] {};
\node[] (n21) [below=of n17] {};
\node[] (n22) [right=of n21] {};
\node[] (n23) [right=of n22] {};
\node[] (n24) [right=of n23] {};

\draw[-] (n3.east) -- (n4.west);
\draw[-] (n11.east) -- (n4.west);
\draw[-] (n15.east) -- (n4.west);
\draw[-] (n23.east) -- (n4.west);
\draw[-] (n3.east) -- (n8.west);
\draw[-] (n11.east) -- (n8.west);
\draw[-] (n15.east) -- (n8.west);
\draw[-] (n23.east) -- (n8.west);
\draw[-] (n3.east) -- (n16.west);
\draw[-] (n11.east) -- (n16.west);
\draw[-] (n15.east) -- (n16.west);
\draw[-] (n23.east) -- (n16.west);
\draw[-] (n3.east) -- (n24.west);
\draw[-] (n11.east) -- (n24.west);
\draw[-] (n15.east) -- (n24.west);
\draw[-] (n23.east) -- (n24.west);
\draw[-] (n3.east) -- (n20.west);
\draw[-] (n11.east) -- (n20.west);
\draw[-] (n15.east) -- (n20.west);
\draw[-] (n19.east) -- (n20.west);
\draw[-] (n23.east) -- (n20.west);

\draw[-] (n18.east) -- (n11.west);
\draw[-] (n18.east) -- (n15.west);
\draw[-] (n18.east) -- (n19.west);
\draw[-] (n18.east) -- (n23.west);
\draw[-] (n2.east) -- (n3.west);
\draw[-] (n10.east) -- (n3.west);
\draw[-] (n14.east) -- (n3.west);
\draw[-] (n22.east) -- (n3.west);

\draw[-] (n2.east) -- (n11.west);
\draw[-] (n10.east) -- (n11.west);
\draw[-] (n14.east) -- (n11.west);
\draw[-] (n22.east) -- (n15.west);
\draw[-] (n10.east) -- (n15.west);
\draw[-] (n14.east) -- (n15.west);
\draw[-] (n2.east) -- (n19.west);
\draw[-] (n10.east) -- (n19.west);
\draw[-] (n14.east) -- (n19.west);
\draw[-] (n22.east) -- (n19.west);
\draw[-] (n2.east) -- (n23.west);
\draw[-] (n10.east) -- (n23.west);
\draw[-] (n22.east) -- (n23.west);

\draw[-] (n6.east) -- (n3.west);
\draw[-] (n6.east) -- (n11.west);
\draw[-] (n6.east) -- (n15.west);
\draw[-] (n6.east) -- (n23.west);

\draw[-] (n13.east) -- (n18.west);
\draw[-] (n1.east) -- (n2.west);
\draw[-] (n5.east) -- (n2.west);
\draw[-] (n17.east) -- (n2.west);
\draw[-] (n21.east) -- (n2.west);
\draw[-] (n1.east) -- (n6.west);
\draw[-] (n5.east) -- (n6.west);
\draw[-] (n17.east) -- (n6.west);
\draw[-] (n21.east) -- (n6.west);
\draw[-] (n1.east) -- (n10.west);
\draw[-] (n5.east) -- (n10.west);
\draw[-] (n17.east) -- (n10.west);
\draw[-] (n21.east) -- (n10.west);
\draw[-] (n1.east) -- (n14.west);
\draw[-] (n5.east) -- (n14.west);
\draw[-] (n17.east) -- (n14.west);
\draw[-] (n21.east) -- (n14.west);
\draw[-] (n1.east) -- (n18.west);
\draw[-] (n5.east) -- (n18.west);
\draw[-] (n17.east) -- (n18.west);
\draw[-] (n21.east) -- (n18.west);
\draw[-] (n1.east) -- (n22.west);
\draw[-] (n5.east) -- (n22.west);
\draw[-] (n17.east) -- (n22.west);
\draw[-] (n21.east) -- (n22.west);

\draw[-] (n13.east) -- (n14.west);
\draw[-] (n13.east) -- (n22.west);
\draw[-] (n14.east) -- (n7.west);
\draw[-] (n22.east) -- (n7.west);

\draw[-] (n9.east) -- (n2.west);
\draw[-] (n9.east) -- (n6.west);
\draw[-] (n9.east) -- (n10.west);
\draw[-] (n2.east) -- (n7.west);
\draw[-] (n10.east) -- (n7.west);
\draw[-] (n6.east) -- (n7.west);
\draw[-, green!75!black, very thick] (n9.east) -- (n2.west);
\draw[-, green!75!black, very thick] (n9.east) -- (n6.west);
\draw[-, green!75!black, very thick] (n9.east) -- (n10.west);
\draw[-, green!75!black, very thick] (n2.east) -- (n7.west);
\draw[-, green!75!black, very thick] (n10.east) -- (n7.west);
\draw[-, green!75!black, very thick] (n6.east) -- (n7.west);

\draw[-] (n7.east) -- (n4.west);
\draw[-] (n7.east) -- (n8.west);
\draw[-] (n7.east) -- (n16.west);
\draw[-] (n7.east) -- (n20.west);
\draw[-] (n7.east) -- (n24.west);

\input{assets/decide/setup}

\end{tikzpicture}
  \caption{L1b indirectly committed ($2f+1$ supports via anchor L3a)}
  \label{fig:decide5}
\end{subfigure}
\hfill
\begin{subfigure}[t]{0.49\textwidth}
  \centering
  \begin{tikzpicture}[scale=0.6, transform shape, node distance=0.75cm and 2.25cm,
    every node/.style={draw, minimum size=1cm}, circle]

\node[] (n1) at (0, 0) {};
\node[] (n2) [right=of n1] {};
\node[] (n3) [right=of n2] {};
\node[fill=gray!50] (n4) [right=of n3] {L4b};
\node[] (n5) [below=of n1] {};
\node[fill=green!75!black] (n6) [right=of n5] {L2b};
\node[fill=green!75!black, draw=orange, ultra thick] (n7) [right=of n6] {L3a};
\node[] (n8) [right=of n7] {};
\node[fill=green!75!black] (n9) [below=of n5] {L1b};
\node[] (n10) [right=of n9] {};
\node[] (n11) [right=of n10] {};
\node[draw=none] (n12) [right=of n11] {};
\node[fill=red!90!black, draw=black, ultra thick] (n13) [below=of n9] {L1a};
\node[draw=red!90!black, very thick] (n14) [right=of n13] {};
\node[] (n15) [right=of n14] {};
\node[] (n16) [right=of n15] {};
\node[] (n17) [below=of n13] {};
\node[fill=gray!50] (n18) [right=of n17] {L2a};
\node[fill=red!90!black] (n19) [right=of n18] {L3b};
\node[] (n20) [right=of n19] {};
\node[] (n21) [below=of n17] {};
\node[draw=red!90!black, very thick] (n22) [right=of n21] {};
\node[] (n23) [right=of n22] {};
\node[] (n24) [right=of n23] {};

\draw[-] (n3.east) -- (n4.west);
\draw[-] (n11.east) -- (n4.west);
\draw[-] (n15.east) -- (n4.west);
\draw[-] (n23.east) -- (n4.west);
\draw[-] (n3.east) -- (n8.west);
\draw[-] (n11.east) -- (n8.west);
\draw[-] (n15.east) -- (n8.west);
\draw[-] (n23.east) -- (n8.west);
\draw[-] (n3.east) -- (n16.west);
\draw[-] (n11.east) -- (n16.west);
\draw[-] (n15.east) -- (n16.west);
\draw[-] (n23.east) -- (n16.west);
\draw[-] (n3.east) -- (n24.west);
\draw[-] (n11.east) -- (n24.west);
\draw[-] (n15.east) -- (n24.west);
\draw[-] (n23.east) -- (n24.west);
\draw[-] (n3.east) -- (n20.west);
\draw[-] (n11.east) -- (n20.west);
\draw[-] (n15.east) -- (n20.west);
\draw[-] (n19.east) -- (n20.west);
\draw[-] (n23.east) -- (n20.west);

\draw[-] (n18.east) -- (n11.west);
\draw[-] (n18.east) -- (n15.west);
\draw[-] (n18.east) -- (n19.west);
\draw[-] (n18.east) -- (n23.west);
\draw[-] (n2.east) -- (n3.west);
\draw[-] (n10.east) -- (n3.west);
\draw[-] (n14.east) -- (n3.west);
\draw[-] (n22.east) -- (n3.west);

\draw[-] (n2.east) -- (n11.west);
\draw[-] (n10.east) -- (n11.west);
\draw[-] (n14.east) -- (n11.west);
\draw[-] (n22.east) -- (n15.west);
\draw[-] (n10.east) -- (n15.west);
\draw[-] (n14.east) -- (n15.west);
\draw[-] (n2.east) -- (n19.west);
\draw[-] (n10.east) -- (n19.west);
\draw[-] (n14.east) -- (n19.west);
\draw[-] (n22.east) -- (n19.west);
\draw[-] (n2.east) -- (n23.west);
\draw[-] (n10.east) -- (n23.west);
\draw[-] (n22.east) -- (n23.west);

\draw[-] (n7.east) -- (n4.west);
\draw[-] (n7.east) -- (n8.west);
\draw[-] (n7.east) -- (n16.west);
\draw[-] (n7.east) -- (n20.west);
\draw[-] (n7.east) -- (n24.west);
\draw[-] (n6.east) -- (n3.west);
\draw[-] (n6.east) -- (n11.west);
\draw[-] (n6.east) -- (n15.west);
\draw[-] (n6.east) -- (n23.west);

\draw[-] (n13.east) -- (n18.west);
\draw[-] (n1.east) -- (n2.west);
\draw[-] (n5.east) -- (n2.west);
\draw[-] (n17.east) -- (n2.west);
\draw[-] (n21.east) -- (n2.west);
\draw[-] (n1.east) -- (n6.west);
\draw[-] (n5.east) -- (n6.west);
\draw[-] (n17.east) -- (n6.west);
\draw[-] (n21.east) -- (n6.west);
\draw[-] (n1.east) -- (n10.west);
\draw[-] (n5.east) -- (n10.west);
\draw[-] (n17.east) -- (n10.west);
\draw[-] (n21.east) -- (n10.west);
\draw[-] (n1.east) -- (n14.west);
\draw[-] (n5.east) -- (n14.west);
\draw[-] (n17.east) -- (n14.west);
\draw[-] (n21.east) -- (n14.west);
\draw[-] (n1.east) -- (n18.west);
\draw[-] (n5.east) -- (n18.west);
\draw[-] (n17.east) -- (n18.west);
\draw[-] (n21.east) -- (n18.west);
\draw[-] (n1.east) -- (n22.west);
\draw[-] (n5.east) -- (n22.west);
\draw[-] (n17.east) -- (n22.west);
\draw[-] (n21.east) -- (n22.west);

\draw[-] (n9.east) -- (n2.west);
\draw[-] (n9.east) -- (n6.west);
\draw[-] (n9.east) -- (n10.west);
\draw[-] (n2.east) -- (n7.west);
\draw[-] (n10.east) -- (n7.west);
\draw[-] (n6.east) -- (n7.west);

\draw[-, red!90!black, very thick] (n13.east) -- (n14.west);
\draw[-, red!90!black, very thick] (n13.east) -- (n22.west);
\draw[-, red!90!black, very thick] (n14.east) -- (n7.west);
\draw[-, red!90!black, very thick] (n22.east) -- (n7.west);

\input{assets/decide/setup}

\end{tikzpicture}
  \caption{L1a indirectly skipped (insufficient supports via anchor L3a)}
  \label{fig:decide6}
\end{subfigure}

\caption{Walkthrough of the \lmysticeti decision rule}
\label{fig:decision}
\end{figure}

\section{Early Block Production Optimization}\label{sec:optimization}

Undecided leader blocks slow down protocol performance by requiring the indirect rule, which takes additional rounds to decide a leader, and by terminating the leader sequence as described in Step 3. This optimization reduces undecided leaders by letting validators produce blocks before the timeout expires.

The standard protocol requires validators to wait up to $2 \cdot \Delta$ for all leader blocks before timing out and broadcasting a block for the next round. When a leader is slow, validators may need to wait the full timeout. Since validators do not proceed in exact lockstep but enter rounds at different times, some validators will timeout and produce blocks for the next round before others who remain waiting on the slow leader.

The optimization allows validators to produce blocks early after receiving $2f + 1$ blames for a slow leader from other validators. With $2f + 1$ blames, it is impossible for the slow leader to achieve the $4f + 1$ supports required for direct commitment. At least $f + 1$ of the blames originate from honest validators, leaving only $3f$ honest validators and $f$ dishonest validators, who can equivocate, to produce blocks supporting the slow leader. However, $3f + f = 4f$ supports is insufficient for direct commitment.

Consider a scenario where a validator has received all the necessary information to produce the next block except for a single current leader block. Under the standard protocol, this validator must wait the full timeout period despite being otherwise ready. With the optimization, once $2f + 1$ blames are seen, the validator can immediately produce the block. The slow leader will either be directly skipped if it accumulates $4f + 1$ blames or remain undecided, but neither scenario impacts the safety of the protocol.

This optimization reduces latency by encouraging direct skips over undecided outcomes after the direct decision rule. The benefits are especially helpful for crashed leaders, which are more common than Byzantine faults in practice \cite{mysticeti}. When a leader crashes and will never produce a block, only $2f + 1$ validators must wait the full timeout before producing blames, while the remaining validators will proceed immediately afterwards.

The optimization preserves safety guarantees as slow leaders that eventually produce valid blocks remain included in the total ordering. If a slow leader eventually produces a block, it will be included because the next block that the validator produces will link to it and eventually find its way into the causal history of a committed leader, ensuring no valid transactions are lost. Figure~\ref{fig:optimization} illustrates the full process of the optimization.

\captionsetup[subfigure]{justification=centering, skip=-0.25cm}
\begin{figure}[htbp]
\centering

\begin{subfigure}[t]{0.32\textwidth}
    \centering
    \begin{tikzpicture}[scale=0.8, transform shape, node distance=0.75cm and 2.25cm,
    every node/.style={draw, minimum size=1cm}, circle]

\node[] (n3) at (0, 0) {};
\node[draw=none] (n4) [right=of n3] {};
\node[] (n7) [below=of n3] {};
\node[draw=none] (n8) [right=of n7] {};
\node[] (n11) [below=of n7] {};
\node[dashed] (n12) [right=of n11] {?};
\node[] (n15) [below=of n11] {};
\node[draw=none] (n16) [right=of n15] {};
\node[dotted] (n19) [below=of n15] {L};
\node[draw=none] (n20) [right=of n19] {};
\node[] (n23) [below=of n19] {};
\node[draw=none] (n24) [right=of n23] {};

\node[draw=none, anchor=east] at (n3.west) {V1};
\node[draw=none, anchor=east] at (n7.west) {V2};
\node[draw=none, anchor=east] at (n11.west) {V3};
\node[draw=none, anchor=east] at (n15.west) {V4};
\node[draw=none, anchor=east] at (n19.west) {V5};
\node[draw=none, anchor=east] at (n23.west) {V6};

\end{tikzpicture}
    \caption{V3 waits for slow leader V5}
    \label{fig:opti1}
\end{subfigure}
\hfill
\begin{subfigure}[t]{0.32\textwidth}
    \centering
    \begin{tikzpicture}[scale=0.8, transform shape, node distance=0.75cm and 2.25cm,
    every node/.style={draw, minimum size=1cm}, circle]

\node[] (n3) at (0, 0) {};
\node[] (n4) [right=of n3] {};
\node[] (n7) [below=of n3] {};
\node[] (n8) [right=of n7] {};
\node[] (n11) [below=of n7] {};
\node[dashed] (n12) [right=of n11] {?};
\node[] (n15) [below=of n11] {};
\node[draw=none] (n16) [right=of n15] {};
\node[dotted] (n19) [below=of n15] {L};
\node[draw=none] (n20) [right=of n19] {};
\node[] (n23) [below=of n19] {};
\node[] (n24) [right=of n23] {};

\draw[-] (n3.east) -- (n4.west);
\draw[-] (n7.east) -- (n4.west);
\draw[-] (n11.east) -- (n4.west);
\draw[-] (n23.east) -- (n4.west);
\draw[-] (n3.east) -- (n8.west);
\draw[-] (n7.east) -- (n8.west);
\draw[-] (n11.east) -- (n8.west);
\draw[-] (n15.east) -- (n8.west);
\draw[-] (n23.east) -- (n8.west);
\draw[-] (n3.east) -- (n24.west);
\draw[-] (n11.east) -- (n24.west);
\draw[-] (n15.east) -- (n24.west);
\draw[-] (n23.east) -- (n24.west);
\draw[-] (n7.east) -- (n24.west);

\node[draw=none, anchor=east] at (n3.west) {V1};
\node[draw=none, anchor=east] at (n7.west) {V2};
\node[draw=none, anchor=east] at (n11.west) {V3};
\node[draw=none, anchor=east] at (n15.west) {V4};
\node[draw=none, anchor=east] at (n19.west) {V5};
\node[draw=none, anchor=east] at (n23.west) {V6};

\end{tikzpicture}
    \caption{V1, V2, V6 timeout and blame V5}
    \label{fig:opti2}
\end{subfigure}
\hfill
\begin{subfigure}[t]{0.32\textwidth}
    \centering
    \begin{tikzpicture}[scale=0.8, transform shape, node distance=0.75cm and 2.25cm,
    every node/.style={draw, minimum size=1cm}, circle]

\node[] (n3) at (0, 0) {};
\node[] (n4) [right=of n3] {};
\node[] (n7) [below=of n3] {};
\node[] (n8) [right=of n7] {};
\node[] (n11) [below=of n7] {};
\node[] (n12) [right=of n11] {};
\node[] (n15) [below=of n11] {};
\node[draw=none] (n16) [right=of n15] {};
\node[dotted] (n19) [below=of n15] {L};
\node[draw=none] (n20) [right=of n19] {};
\node[] (n23) [below=of n19] {};
\node[] (n24) [right=of n23] {};

\draw[-] (n3.east) -- (n4.west);
\draw[-] (n7.east) -- (n4.west);
\draw[-] (n11.east) -- (n4.west);
\draw[-] (n23.east) -- (n4.west);
\draw[-] (n3.east) -- (n8.west);
\draw[-] (n7.east) -- (n8.west);
\draw[-] (n11.east) -- (n8.west);
\draw[-] (n15.east) -- (n8.west);
\draw[-] (n23.east) -- (n8.west);
\draw[-] (n3.east) -- (n12.west);
\draw[-] (n7.east) -- (n12.west);
\draw[-] (n11.east) -- (n12.west);
\draw[-] (n15.east) -- (n12.west);
\draw[-] (n23.east) -- (n12.west);
\draw[-] (n3.east) -- (n24.west);
\draw[-] (n11.east) -- (n24.west);
\draw[-] (n15.east) -- (n24.west);
\draw[-] (n23.east) -- (n24.west);
\draw[-] (n7.east) -- (n24.west);

\node[draw=none, anchor=east] at (n3.west) {V1};
\node[draw=none, anchor=east] at (n7.west) {V2};
\node[draw=none, anchor=east] at (n11.west) {V3};
\node[draw=none, anchor=east] at (n15.west) {V4};
\node[draw=none, anchor=east] at (n19.west) {V5};
\node[draw=none, anchor=east] at (n23.west) {V6};

\end{tikzpicture}
    \caption{V3 proceeds after observing $2f+1$ blames}
    \label{fig:opti3}
\end{subfigure}

\vspace{0.5cm}

\begin{subfigure}[t]{0.32\textwidth}
    \centering
    \begin{tikzpicture}[scale=0.8, transform shape, node distance=0.75cm and 2.25cm,
    every node/.style={draw, minimum size=1cm}, circle]

\node[] (n3) at (0, 0) {};
\node[] (n4) [right=of n3] {};
\node[] (n7) [below=of n3] {};
\node[] (n8) [right=of n7] {};
\node[] (n11) [below=of n7] {};
\node[] (n12) [right=of n11] {};
\node[] (n15) [below=of n11] {};
\node[draw=none] (n16) [right=of n15] {};
\node[fill=gray!50] (n19) [below=of n15] {L};
\node[draw=none] (n20) [right=of n19] {};
\node[] (n23) [below=of n19] {};
\node[] (n24) [right=of n23] {};

\draw[-] (n3.east) -- (n4.west);
\draw[-] (n7.east) -- (n4.west);
\draw[-] (n11.east) -- (n4.west);
\draw[-] (n23.east) -- (n4.west);
\draw[-] (n3.east) -- (n8.west);
\draw[-] (n7.east) -- (n8.west);
\draw[-] (n11.east) -- (n8.west);
\draw[-] (n15.east) -- (n8.west);
\draw[-] (n23.east) -- (n8.west);
\draw[-] (n3.east) -- (n12.west);
\draw[-] (n7.east) -- (n12.west);
\draw[-] (n11.east) -- (n12.west);
\draw[-] (n15.east) -- (n12.west);
\draw[-] (n23.east) -- (n12.west);
\draw[-] (n3.east) -- (n24.west);
\draw[-] (n11.east) -- (n24.west);
\draw[-] (n15.east) -- (n24.west);
\draw[-] (n23.east) -- (n24.west);
\draw[-] (n7.east) -- (n24.west);

\node[draw=none, anchor=east] at (n3.west) {V1};
\node[draw=none, anchor=east] at (n7.west) {V2};
\node[draw=none, anchor=east] at (n11.west) {V3};
\node[draw=none, anchor=east] at (n15.west) {V4};
\node[draw=none, anchor=east] at (n19.west) {V5};
\node[draw=none, anchor=east] at (n23.west) {V6};

\end{tikzpicture}
    \caption{V5's block arrives late}
    \label{fig:opti4}
\end{subfigure}
\hfill
\begin{subfigure}[t]{0.32\textwidth}
    \centering
    \begin{tikzpicture}[scale=0.8, transform shape, node distance=0.75cm and 2.25cm,
    every node/.style={draw, minimum size=1cm}, circle]

\node[] (n3) at (0, 0) {};
\node[] (n4) [right=of n3] {};
\node[] (n7) [below=of n3] {};
\node[] (n8) [right=of n7] {};
\node[] (n11) [below=of n7] {};
\node[] (n12) [right=of n11] {};
\node[] (n15) [below=of n11] {};
\node[] (n16) [right=of n15] {};
\node[fill=red!90!black] (n19) [below=of n15] {L};
\node[] (n20) [right=of n19] {};
\node[] (n23) [below=of n19] {};
\node[] (n24) [right=of n23] {};

\draw[-] (n3.east) -- (n4.west);
\draw[-] (n7.east) -- (n4.west);
\draw[-] (n11.east) -- (n4.west);
\draw[-] (n23.east) -- (n4.west);
\draw[-] (n3.east) -- (n8.west);
\draw[-] (n7.east) -- (n8.west);
\draw[-] (n11.east) -- (n8.west);
\draw[-] (n15.east) -- (n8.west);
\draw[-] (n23.east) -- (n8.west);
\draw[-] (n3.east) -- (n12.west);
\draw[-] (n7.east) -- (n12.west);
\draw[-] (n11.east) -- (n12.west);
\draw[-] (n15.east) -- (n12.west);
\draw[-] (n23.east) -- (n12.west);
\draw[-] (n3.east) -- (n16.west);
\draw[-] (n7.east) -- (n16.west);
\draw[-] (n11.east) -- (n16.west);
\draw[-] (n15.east) -- (n16.west);
\draw[-] (n23.east) -- (n16.west);
\draw[-] (n3.east) -- (n20.west);
\draw[-] (n7.east) -- (n20.west);
\draw[-] (n11.east) -- (n20.west);
\draw[-] (n15.east) -- (n20.west);
\draw[-] (n23.east) -- (n20.west);
\draw[-] (n3.east) -- (n24.west);
\draw[-] (n11.east) -- (n24.west);
\draw[-] (n15.east) -- (n24.west);
\draw[-] (n23.east) -- (n24.west);
\draw[-] (n7.east) -- (n24.west);

\draw[-, red!90!black, thick] (n19.east) -- (n20.west);

\node[draw=none, anchor=east] at (n3.west) {V1};
\node[draw=none, anchor=east] at (n7.west) {V2};
\node[draw=none, anchor=east] at (n11.west) {V3};
\node[draw=none, anchor=east] at (n15.west) {V4};
\node[draw=none, anchor=east] at (n19.west) {V5};
\node[draw=none, anchor=east] at (n23.west) {V6};

\end{tikzpicture}
    \caption{V5 directly skipped with insufficient support}
    \label{fig:opti5}
\end{subfigure}
\hfill
\begin{subfigure}[t]{0.32\textwidth}
    \centering
    \begin{tikzpicture}[scale=0.8, transform shape, node distance=0.75cm and 2.25cm,
    every node/.style={draw, minimum size=1cm}, circle]

\node[] (n3) at (0, 0) {};
\node[] (n4) [right=of n3] {};
\node[] (n7) [below=of n3] {};
\node[] (n8) [right=of n7] {};
\node[] (n11) [below=of n7] {};
\node[] (n12) [right=of n11] {};
\node[] (n15) [below=of n11] {};
\node[] (n16) [right=of n15] {};
\node[fill=gray!50] (n19) [below=of n15] {L};
\node[] (n20) [right=of n19] {};
\node[] (n23) [below=of n19] {};
\node[] (n24) [right=of n23] {};

\draw[-] (n3.east) -- (n4.west);
\draw[-] (n7.east) -- (n4.west);
\draw[-] (n11.east) -- (n4.west);
\draw[-] (n23.east) -- (n4.west);
\draw[-] (n3.east) -- (n8.west);
\draw[-] (n7.east) -- (n8.west);
\draw[-] (n11.east) -- (n8.west);
\draw[-] (n15.east) -- (n8.west);
\draw[-] (n23.east) -- (n8.west);
\draw[-] (n3.east) -- (n12.west);
\draw[-] (n7.east) -- (n12.west);
\draw[-] (n11.east) -- (n12.west);
\draw[-] (n15.east) -- (n12.west);
\draw[-] (n23.east) -- (n12.west);
\draw[-] (n3.east) -- (n16.west);
\draw[-] (n7.east) -- (n16.west);
\draw[-] (n11.east) -- (n16.west);
\draw[-] (n15.east) -- (n16.west);
\draw[-] (n23.east) -- (n16.west);
\draw[-] (n3.east) -- (n20.west);
\draw[-] (n7.east) -- (n20.west);
\draw[-] (n11.east) -- (n20.west);
\draw[-] (n15.east) -- (n20.west);
\draw[-] (n19.east) -- (n20.west);
\draw[-] (n23.east) -- (n20.west);
\draw[-] (n3.east) -- (n24.west);
\draw[-] (n11.east) -- (n24.west);
\draw[-] (n15.east) -- (n24.west);
\draw[-] (n23.east) -- (n24.west);
\draw[-] (n7.east) -- (n24.west);

\draw[dashed, red!90!black, thick] (n19.east) -- (n12.west);

\node[draw=none, anchor=east] at (n3.west) {V1};
\node[draw=none, anchor=east] at (n7.west) {V2};
\node[draw=none, anchor=east] at (n11.west) {V3};
\node[draw=none, anchor=east] at (n15.west) {V4};
\node[draw=none, anchor=east] at (n19.west) {V5};
\node[draw=none, anchor=east] at (n23.west) {V6};

\end{tikzpicture}
    \caption{Without optimization: V5 remains undecided}
    \label{fig:opti6}
\end{subfigure}

\caption{Walkthrough of the early block production optimization}
\label{fig:optimization}
\end{figure}
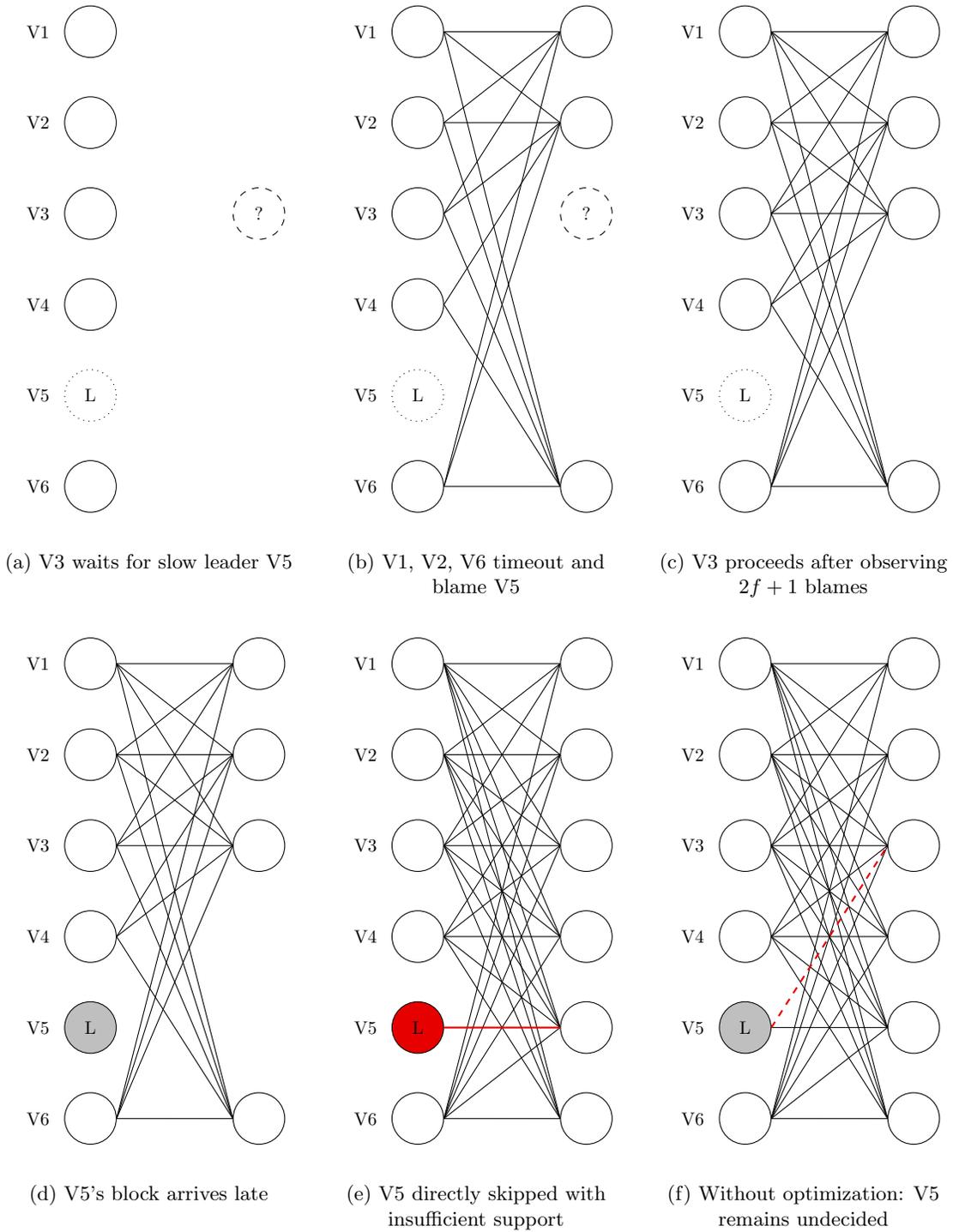

\section{Algorithms}\label{sec:algos}

\lmysticeti is formally presented with three algorithms that outline the consensus protocol. First, Algorithm~\ref{alg:helper} provides essential helper functions that implement common operations used in many DAG-based consensus protocols \cite{mysticeti, bullshark, narwhal, dag-rider}. Next, Algorithm~\ref{alg:decider} defines the Decider instance, which every validator creates for every leader slot. Finally, Algorithm~\ref{alg:lmys} presents the main \lmysticeti protocol that orchestrates the overall consensus process and is invoked whenever a new block is received.

\clearpage

Algorithm~\ref{alg:helper} establishes helper functions required by the consensus protocol. Of note, \linebreak \textsc{GetLeaderBlocks} accounts for potential leader equivocation by returning all blocks authored by the input leader and \textsc{GetPredefinedLeader} implements a round robin leader schedule. Under this schedule, each validator begins as the lowest-ranked leader and incrementally increases in rank until reaching the highest position, after which that validator is excluded from being a leader until the subsequent cycle. This uniform rotation among validators formalizes a deterministic ordering based on round and rank.

\begin{algorithm}[ht]
\caption{Helper Functions}\label{alg:helper}
\begin{algorithmic}[1]

\State \texttt{$V$} \Comment{The set of validators}
\Statex

\Function{GetDecisionBlocks}{$w$}
    \State $r_{decision} \gets $\Call{DecisionRound}{$w$}
    \State \Return $DAG[r_{decision}]$
\EndFunction
\Statex

\Function{GetPredefinedLeader}{$r, rank$} \Comment{Rank 0 is highest; round-robin schedule}
    \State \Return $V[(r + rank) \bmod |V|]$ 
\EndFunction
\Statex

\Function{GetLeaderBlocks}{$w, rank$} \Comment{Leader may equivocate with multiple blocks}
    \State $r_{propose} \gets $\Call{ProposeRound}{$w$}
    \State $leader \gets $\Call{GetPredefinedLeader}{$r_{propose}, rank$}
    \State \Return $\{ b \in DAG[r_{propose}] : b.author = leader\}$
\EndFunction
\Statex

\Function{IsSupport}{$b_{support}, b_{leader}$}
    \State \Return $b_{leader} \in b_{support}.parents$
\EndFunction
\Statex

\Function{Link}{$b_{old}, b_{new}$}
    \State \Return $\exists$ a sequence of $k \in \mathbb{N}$ blocks $b_1, \ldots, b_k$ s.t.
    \Statex \hspace{2em} $b_1 = b_{old} \land b_k = b_{new} \land \forall j \in [2, k] : b_j \in \bigcup_{r \geq 1} DAG[r] \land b_{j-1} \in b_j.parents$
\EndFunction

\end{algorithmic}
\end{algorithm}

Algorithm~\ref{alg:decider} implements the core decision-making logic for consensus participants to evaluate a leader slot. Each Decider instance is parameterized by wave and leader offsets, enabling the protocol to pipeline waves by running them in parallel. \textsc{SupportedLeader} and \textsc{SkippedLeader} implement the $4f + 1$ threshold that ensures safety. \textsc{TryDirectDecide} attempts a direct decision by checking whether a leader has sufficient support or blame. \textsc{TryIndirectDecide} searches the DAG to find an anchor, sees if that anchor is committed, and if so, determines whether the current leader has a thick link to the anchor. \textsc{ThickLink} verifies the indirect decision rule by counting decision blocks that both support the current leader and are connected to the anchor, returning true when this count meets the $2f + 1$ threshold.

\begin{algorithm}[ht]
\caption{Decider Instance}\label{alg:decider}
\begin{algorithmic}[1]

\State waveOffset \Comment{The first parameter of the Decider (i)}
\State leaderOffset \Comment{The second parameter of the Decider (l)}
\State waveLength $= 2$
\Statex

\Function{WaveNumber}{$r$}
    \State \Return $(r - \text{waveOffset}) / \text{waveLength}$ 
\EndFunction
\Statex

\Function{ProposeRound}{$w$}
    \State \Return $(w * \text{waveLength}) + \text{waveOffset}$ 
\EndFunction
\Statex

\Function{DecisionRound}{$w$}
    \State \Return \Call{ProposeRound}{$w$}$ + (\text{waveLength} - 1)$
\EndFunction
\Statex

\Function{SupportedLeader}{$w, b_{leader}$}
    \State $B_{decision} \gets$ \Call{GetDecisionBlocks}{$w$}
    \State \Return $|\{ b' \in B_{decision} : $ \Call{IsSupport}{$b', b_{leader}$}$ \}| \geq 4f + 1$
\EndFunction
\Statex

\Function{SkippedLeader}{$w, b_{leader}$}
    \State $B_{decision} \gets$ \Call{GetDecisionBlocks}{$w$}
    \State \Return $|\{ b' \in B_{decision} : \neg$\Call{IsSupport}{$b', b_{leader}$}$ \}| \geq 4f + 1$
\EndFunction
\Statex

\Function{TryDirectDecide}{$w$}
    \State $B_{leader} \gets$ \Call{GetLeaderBlocks}{$w$, leaderOffset}
    \For{$b_{leader} \in B_{leader}$}
        \If{\Call{SkippedLeader}{$w, b_{leader}$}} \Return \texttt{Skip}
        \EndIf
        \If{\Call{SupportedLeader}{$w, b_{leader}$}} \Return \texttt{Commit}$(b_{leader})$
        \EndIf
    \EndFor
    \State \Return $\bot$
\EndFunction
\Statex

\Function{ThickLink}{$b_{anchor}, b_{leader}$}
    \State $w \gets$ \Call{WaveNumber}{$b_{leader}.round$}
    \State $B_{decision} \gets$ \Call{GetDecisionBlocks}{$w$}
    \State \Return $|\{ b \in B_{decision} : $ \Call{IsSupport}{$b, b_{leader}$} $\land$ \Call{Link}{$b, b_{anchor}$}$ \}| \geq 2f + 1$
\EndFunction
\Statex

\Function{TryIndirectDecide}{$w, S$}
    \State $r_{decision} \gets $\Call{DecisionRound}{$w$}
    \State $s_{anchor} \gets$ first $s \in S$ s.t. $r_{decision} < s.round \land s \neq$ \texttt{Skip}
    \If{$s_{anchor} = \texttt{Commit}(b_{anchor})$}
        \State $B_{leader} \gets$ \Call{GetLeaderBlocks}{$w$, leaderOffset}
        \If{$\exists \; b_{leader} \in B_{leader}$ s.t. \Call{ThickLink}{$b_{anchor}, b_{leader}$}} \Return \texttt{Commit}$(b_{leader})$
        \Else \; \Return \texttt{Skip}
        \EndIf
    \EndIf
    \State \Return $\bot$
\EndFunction

\end{algorithmic}
\end{algorithm}

\clearpage

Algorithm~\ref{alg:lmys} orchestrates the overall consensus process by systematically examining potential decisions. The protocol operates with a \texttt{waveLength} of 2 rounds and supports between 1 and $4f + 1$, the number of honest validators, leaders per round as specified by the \texttt{leadersPerRound} parameter. \textsc{TryDecide} iterates through rounds and ranks in reverse order, as rank 0 is the highest rank in a round. It creates appropriate Decider instances and attempts both direct and, as needed, indirect decisions for each leader. \textsc{ExtendCommitSequence} processes the resulting decision sequence, keeping committed blocks until an undecided leader is encountered and invoking the linearization process to produce the total ordering output.

\begin{algorithm}[ht]
\caption{\lmysticeti}\label{alg:lmys}
\begin{algorithmic}[1]

\State leadersPerRound \Comment{A number between 1 and $4f + 1$}
\State waveLength $= 2$
\Statex

\Function{TryDecide}{$r_{committed}, r_{highest}$}
    \State $S \gets [ \; ]$
    \For{$r \gets [r_{highest}$ \textbf{down to} $r_{committed} + 1$]}
        \For{$l \gets \text{[leadersPerRound} - 1$ \textbf{down to} $0$]}
            \State $i \gets r \; \bmod$ waveLength
            \State $D \gets$ \Call{Decider}{$i, l$}
            \State $w \gets D.$\Call{WaveNumber}{$r$}
            \State $status \gets D.$\Call{TryDirectDecide}{$w$}
            \If{$status = \bot$} $status \gets D.$\Call{TryIndirectDecide}{$w, S$}
            \EndIf
            \State $S \gets status \parallel S$
        \EndFor
    \EndFor
    \State \Return $S$
\EndFunction
\Statex

\Function{ExtendCommitSequence}{$r_{committed}, r_{highest}$}
    \State $S \gets$ \Call{TryDecide}{$r_{committed}, r_{highest}$}
    \State $S_{commit} \gets [ \; ]$
    \For{$status \in S$}
        \If{$status = \bot$} \textbf{break}
        \EndIf
        \If{$status = \texttt{Commit}(b_{leader})$} $S_{commit} \gets S_{commit} \parallel b_{leader}$
        \EndIf
    \EndFor
    \State \Return \Call{LinearizeSubDags}{$S_{commit}$}
\EndFunction

\end{algorithmic}
\end{algorithm}

\chapter{Proofs}\label{chap:proofs}

For correctness, both safety and liveness, as defined in Section~\ref{sec:goals}, must be guaranteed. Safety ensures all honest participants reach identical consensus states, while liveness establishes that there is eventual progress toward agreement. Solving BAB is equivalent to guaranteeing safety and liveness.

\section{Safety}\label{sec:security}

\begin{lemma}\label{lem:safety1}
    There will never be a leader block which an honest participant directly commits while another honest participant directly skips.
\end{lemma}
\begin{proof}
    Assume for the sake of contradiction that such a block exists and call it $B$. Thus, there are $4f + 1$ participants which support $B$ and $4f + 1$ participants which blame $B$. Since $f$ participants in the network are Byzantine, there are $3f + 1$ honest participants which support $B$ and a distinct set of $3f + 1$ honest participants which blame $B$. This means there are $(3f + 1) + (3f + 1) = 6f + 2$ honest participants. With the $f$ Byzantine participants, the network contains $(6f + 2) + f = 7f + 2$ total participants. This is a contradiction.
\end{proof}

\begin{lemma}\label{lem:safety2}
    There will never be a leader block which an honest participant directly skips while another honest participant indirectly commits.
\end{lemma}
\begin{proof}
    Assume for the sake of contradiction that such a block exists and call it $B$. Thus, there are $4f + 1$ participants which blame $B$ and $2f + 1$ participants which support $B$. Since $f$ participants in the network are Byzantine, there are $3f + 1$ honest participants which blame $B$ and a distinct set of $f + 1$ honest participants which support $B$. This means there are $(3f + 1) + (f + 1) = 4f + 2$ honest participants. With the $f$ Byzantine participants, the network contains $(4f + 2) + f = 5f + 2$ total participants. This is a contradiction.
\end{proof}

\begin{lemma}\label{lem:safety3}
    If at a round x, 4f + 1 blocks from distinct participants support a block B, then all blocks at future rounds ($>$ x) will link to 2f + 1 supports for B from round x.    
\end{lemma}
\begin{proof}
    Each block links to $4f + 1$ blocks from the previous round. For the sake of contradiction, assume that a block $B'$ in round $r$ ($> x$) does not link to $2f + 1$ supports for $B$ from round $x$. 
    \begin{caseof}
        \case{$r = x + 1$}{$B'$ links to $4f + 1$ blocks from round $x$. Since $4f + 1$ blocks in round $x$ support $B$, by the standard quorum intersection the minimum overlap between $B'$ parents and $B$ supports is $2f + 1$. Thus, the only way for $B'$ to not link to $2f + 1$ supports is if a correct participant equivocated in round $x$. This is a contradiction.}
        \case{$r > x + 1$}{$B'$ links to $4f + 1$ blocks from round $r - 1$. At least $3f + 1$ of these blocks are produced by honest participants. Honest participants always link to their own blocks, which means they will eventually link to their block from round $x + 1$. The above case proves how these blocks from round $x + 1$ link to $2f + 1$ supports for $B$. Thus, the only way for $B'$ to not link to $2f + 1$ supports for $B$ is if all of these honest participants do not link to their own block in round $x + 1$. This is a contradiction.}
    \end{caseof}
\end{proof}

As a result of Lemma \ref{lem:safety3}, we have the following corollary.

\begin{corollary}\label{cor:safety4}
    There will never be a leader block which an honest participant directly commits while another honest participant indirectly skips.
\end{corollary}

\begin{lemma}\label{lem:safety5}
    All honest participants who have decided on a leader block, agree on the decision.
\end{lemma}
\begin{proof}
    Let $B_n$ and $B_m$ be the highest committed leader blocks according to participants $X$ and $Y$ respectively. Without loss of generality, let $n \leq m$. Note that leader blocks decided by $X$ which are higher than $B_n$ are direct skips which according to Lemma \ref{lem:safety1} and Lemma \ref{lem:safety2} will be consistent with $Y$'s decision. The proof continues by induction on the statement for $0 \leq i \leq n$, if both $X$ and $Y$ decide on leader block $B_i$, then they either both commit or both skip the block.

    \proofpart{Base Case $i = n$}{By definition, $X$ directly commits $B_i$ and from Lemma \ref{lem:safety1} and Corollary \ref{cor:safety4}, $Y$ will also commit $B_i$.}
    \proofpart{Induction Step}{Assuming the statement is true regarding $B_i$ for $k + 1 \leq i \leq n$, we prove it is true for $B_k$. This is done by enumerating decision possibilities.
        \begin{enumerate}
            \item If either participant directly commits $B_k$, then by Lemma \ref{lem:safety1} and Corollary \ref{cor:safety4}, the other will commit.
            \item If either participant directly skips $B_k$, then by Lemma \ref{lem:safety1} and Lemma \ref{lem:safety2}, the other will skip.
            \item Both $X$ and $Y$ indirectly decide $B_k$. Let $A_c^X$ and $A_d^Y$ be the anchors used by $X$ and $Y$ to indirectly decide $B_k$. Since $k + 1 < c \leq n$, it follows from the induction hypothesis that $A_c^X = A_d^Y$. Thus, both $X$ and $Y$ use the same anchor to decide $B_k$. The indirect decision rule solely depends on the causal history of the anchor. By using the same anchor, $X$ and $Y$ will agree on the decision for $B_k$.
        \end{enumerate}}
\end{proof}

\begin{theorem}\label{thm:safety}
    \lmysticeti maintains safety.
\end{theorem}
\begin{proof}
    The final state of consensus is a total ordering of all blocks. In \lmysticeti, the total ordering is sustained until an undecided leader block is encountered. From Lemma \ref{lem:safety5}, all honest participants agree upon the decisions of all leader blocks until an undecided leader block. The total ordering of all blocks is a deterministic algorithm run on the sequence of committed leader blocks. Since all honest participants have the same sequence of committed leader blocks, the total ordering will be the same.
\end{proof}

\section{Liveness}\label{sec:liveness}

Liveness is proved in the setting of partial synchrony, fully defined in Section~\ref{sec:network-assumptions}. This model assumes the network eventually becomes synchronous after GST, at which point messages are delivered within a known bound $\Delta$.

\begin{lemma}\label{lem:liveness1}
    All honest participants will, after GST, enter the same round within $\Delta$.
\end{lemma}
\begin{proof}
    Messages sent before GST will deliver in $\Delta$ after GST commences. Thus, the valid block of the highest round that any participant sent before GST will be delivered to all participants in GST + $\Delta$. Upon receiving this block, all honest participants will enter the round.
\end{proof}

\begin{lemma}\label{lem:liveness2}
    Leader blocks from honest participants will, after GST, receive support from all honest participants.
 \end{lemma}
 \begin{proof}
    By Lemma \ref{lem:liveness1}, all honest participants will enter the same round within $\Delta$ after GST. When an honest participant sends its leader block for this round, it is delivered to all participants within $\Delta$. Since the protocol sets the timeout to $2 \cdot \Delta$, any honest participant that has entered the round will receive the leader block and support it before timing out, regardless of when exactly they entered the round relative to other honest participants.
 \end{proof}

As a result of Lemma \ref{lem:liveness2}, we have the following corollary. Recall that there are $4f + 1$ honest participants.

\begin{corollary}\label{cor:liveness3}
    Leader blocks from honest participants will, after GST, be (directly) committed.
\end{corollary}

\begin{lemma}\label{lem:liveness4}
    The round robin schedule of leader block proposers ensures that in a window of $2f + 2$ rounds, there are two consecutive rounds where an honest participant is the proposer of the highest ranked leader block.
\end{lemma}
\begin{proof}
    The network contains $f$ Byzantine participants. In $2f + 2$ rounds, there are $2f + 1$ sets of two consecutive rounds. Due to the schedule being round robin, in at least $f + 1$ of the rounds, an honest participant will be the proposer of the highest ranked leader block. These blocks are the highest ranked leader block in exactly two of the sets. By pigeonhole principle, one set must contain $\lceil \frac{2 * (f + 1)}{2f + 1} \rceil = 2$ honest participants proposing the highest ranked leader block.
\end{proof}

\begin{lemma}\label{lem:liveness5}
    Undecided leader blocks will eventually, after GST, be decided.
\end{lemma}
\begin{proof}
    Let $B$ be an undecided leader block in round $r$. By Lemma \ref{lem:liveness4}, after GST, there will be two consecutive rounds, $j$ and $j + 1$ with $j > r$, where honest participants propose the highest ranked leader block. By Corollary \ref{cor:liveness3}, these highest ranked leader blocks in $j$ and $j + 1$ will be committed. The proof continues by induction on the statement for rounds earlier than $j$, all leader blocks are decided.
    \begin{caseof}
        \proofpart{Base Case}{All undecided leader blocks in rounds $j - 1$ and $j - 2$ will be decided as they now have committed anchors in rounds $j + 1$ and $j$ respectively.}
        \proofpart{Induction Step}{For undecided leader blocks in round $i < j - 2$, $j$ is higher than the decision round of the wave that $i$ is in. Thus, by the induction hypothesis, there are no undecided leader blocks between $i$ and $j$. Hence, the leader block in round $i$ will also decided.}
    \end{caseof}
\end{proof}

\begin{theorem}\label{thm:liveness}
    \lmysticeti maintains liveness.
\end{theorem}
\begin{proof}
    By Lemma \ref{lem:liveness1}, all honest participants synchronize to the same round within $\Delta$ after GST. By Lemma \ref{lem:liveness4}, within any window of $2f + 2$ rounds, there exist two consecutive rounds where honest participants propose the highest ranked leader blocks. By Corollary \ref{cor:liveness3}, these honest leader blocks will be directly committed.
    
    Furthermore, by Lemma \ref{lem:liveness5}, any previously undecided leader blocks will eventually be decided once we have committed blocks from honest participants in consecutive rounds. This creates a cascading effect where the commitment of new honest blocks triggers the decision of older undecided blocks.
    
    Since we can guarantee that consecutive honest leader blocks in the highest rank are committed every $2f + 2$ rounds, and each such instance resolves all pending undecided blocks, the protocol makes continuous progress. No block remains undecided indefinitely and new blocks are regularly committed.
\end{proof}

\section{BAB}\label{sec:bab}
\lmysticeti satisfies the complete BAB specification defined in Section~\ref{sec:bab-def}. In this context, blocks serve as the messages.

\begin{theorem}\label{thm:bab}
    \lmysticeti implements Byzantine Atomic Broadcast.
\end{theorem}
\begin{proof}    
    \textbf{Validity}: If an honest participant broadcasts a block, it will be included in the DAG and processed by the consensus protocol. For leader blocks, Theorem~\ref{thm:liveness} ensures they are eventually decided (either committed or skipped). For non-leader and skipped leader blocks, they are included in the linearization process of committed leader blocks. Thus, all honest participants eventually deliver the same decision regarding all blocks in the system.
    
    \textbf{Agreement}: By Lemma~\ref{lem:safety5}, all honest participants reach identical decisions for all leader blocks. Since the total ordering is constructed deterministically from these decisions, if any honest participant delivers a block in the total ordering, all other honest participants will deliver the same block.
    
    \textbf{Integrity}: The protocol's block validation and cryptographic signatures ensure that blocks are only accepted from their genuine authors. The ordering algorithm processes each decided block exactly once, preventing duplicate delivery in the total ordering.
    
    \textbf{Total Order}: From Theorem~\ref{thm:safety}, all honest participants construct identical total orderings from the same sequence of committed leader blocks. Since all blocks are delivered according to position within these identical orderings, all honest participants deliver blocks in the same sequential order.
\end{proof}

\chapter{Implementation}\label{chap:implementation}

The implementation of the \lmysticeti protocol was written in Rust as a fork of Mysticeti \cite{mysticeti-code}. This allowed for easy integration and reuse of core components from the Mysticeti validator, significantly accelerating development. Specifically, \lmysticeti benefited from using \texttt{tokio} \cite{tokio} with TCP sockets, enabling complete control over asynchronous communication, custom protocol design, and efficient resource management. For asymmetric cryptographic operations, \linebreak \texttt{ed25519-consensus} \cite{ed25519-consensus} was used, while cryptographic hashing was handled with \texttt{blake2} \cite{rustcrypto-hashes}. Both libraries offer strong security guarantees and were well established and maintained. The write ahead log (WAL) provided a lasting record of operations that aided in reasoning about system state. For safe and efficient manipulation of buffer data while minimizing data copying and serialization overhead, the WAL was read using memory-mapped files in combination with \texttt{minibytes} \cite{minibytes}. I/O performance was further improved through the use of vectored writes \cite{dieNetWritev}, which writes to multiple memory regions in a single syscall, streamlining data persistence. By adopting these components without modification, \lmysticeti inherited the performance optimizations already present.

\section{Consensus Protocol}

\lmysticeti's implementation required several modifications to Mysticeti's consensus protocol. First, the existing quorum threshold was adjusted to reflect the updated network size of $5f + 1$ validators and a new threshold was introduced for the indirect decision rule. Second, the wave length was reduced from three to two. Third, the logic related to certificate patterns was removed entirely. Fourth, the direct decision rule was revised to trigger based on accumulated support after a wave's second round. Fifth, the early block production optimization was added. Key code snippets that represent these primary changes are examined below.

In total, the implementation of the \lmysticeti consensus protocol logic required under 350 lines of code difference. These minimal core protocol modifications underscore the reusability and well-architected nature of the Mysticeti codebase, which provided a near plug and play foundation for \lmysticeti. The \lmysticeti implementation is available as open source\footnote{\url{https://github.com/phvv/mysticeti/tree/odontoceti} (commit
e02aeba)}.

\subsection{\normalfont\textsc{SupportedLeader}}

\texttt{enough\_leader\_support} mimics the \textsc{SupportedLeader} function from Algorithm~\ref{alg:decider}. It first retrieves all blocks from the decision round, then iterates through them to count support for the input leader block. Support is determined by checking if each decision block includes the leader block's author in its reference list, which corresponds directly to the \textsc{IsSupport} function. The $4f + 1$ threshold is enforced via collecting validator stakes in \texttt{StakeAggregator<QuorumThreshold>}. If sufficient support is accumulated, the function returns true.

\begin{lstlisting}[language=Rust, basicstyle=\footnotesize\ttfamily, breaklines=true, frame=single]
fn enough_leader_support(&self,
    decision_round: RoundNumber,
    leader_block: &Data<StatementBlock>
) -> bool {
    let decision_blocks = self.block_store.get_blocks_by_round(decision_round);
    let mut support_stake_aggregator = StakeAggregator::<QuorumThreshold>::new();
    for decision_block in &decision_blocks {
        let decider = decision_block.reference().authority;
        if decision_block
            .includes().iter()
            .any(|include| include.authority == leader_block.author_round().0) {
            if support_stake_aggregator.add(decider, &self.committee) {
                return true;
            }
        }
    }
    false
}
\end{lstlisting}

\subsection{\normalfont\textsc{TryIndirectDecide}}

\texttt{decide\_leader\_from\_anchor} implements the indirect decision logic from \textsc{TryIndirectDecide}, also in Algorithm~\ref{alg:decider}, after a committed anchor has been found (lines 30-32). The function begins by retrieving all blocks proposed by the input leader using \texttt{get\_blocks\_at\_authority\_round}, which mirrors \textsc{GetLeaderBlocks} by accounting for potential Byzantine equivocations.

The code's core logic verifies the presence of a thick link between the anchor and a leader block. As in \textsc{ThickLink}, information is gathered in order to collect the required decision blocks. \texttt{potential\_supports} then finds all the decision blocks which are linked to the anchor. The remaining blocks are next checked against each leader block to see if they support it. Support is collected in \texttt{StakeAggregator<IndirectQuorumThreshold>}, which enforces the $2f + 1$ requirement.

If a leader block accumulates enough support in the anchor's causal history, the implementation returns \texttt{LeaderStatus::Commit}; otherwise, it returns \texttt{LeaderStatus::Skip}.

\begin{lstlisting}[language=Rust, basicstyle=\footnotesize\ttfamily, breaklines=true, frame=single]
fn decide_leader_from_anchor(&self,
    anchor: &Data<StatementBlock>,
    leader: AuthorityIndex,
    leader_round: RoundNumber,
) -> LeaderStatus {
    let leader_blocks = self.block_store.get_blocks_at_authority_round(leader, leader_round);
    let wave = self.wave_number(leader_round);
    let decision_round = self.decision_round(wave);
    let decision_blocks = self.block_store.get_blocks_by_round(decision_round);
    let potential_supports: Vec<_> = decision_blocks
        .iter()
        .filter(|block| self.block_store.linked(anchor, block))
        .collect();
    let mut supported_leader_blocks: Vec<_> = leader_blocks
        .into_iter()
        .filter(|_| {
            let mut indirect_support_stake_aggregator =
                StakeAggregator::<IndirectQuorumThreshold>::new();
            for block in &potential_supports {
                let authority = block.reference().authority;
                if block
                    .includes().iter()
                    .any(|include| include.authority == leader) {
                    if indirect_support_stake_aggregator.add(authority, &self.committee) {
                        return true;
                    }
                }
            }
            false
        })
        .collect();
    match supported_leader_blocks.pop() {
        Some(supported_leader_block) => LeaderStatus::Commit(supported_leader_block.clone()),
        None => LeaderStatus::Skip(leader, leader_round),
    }
}
\end{lstlisting}

\subsection{Early Block Optimization}

\texttt{force\_due\_to\_leader\_blames} implements the early block production optimization described in Section~\ref{sec:optimization}. The function determines whether a validator can proceed to produce the next block before the timeout expires by seeing if $2f + 1$ blames have been accumulated for slow leaders from the current round.

The implementation begins by retrieving the next round and verifying that it is high enough. It then identifies leaders from the current round that have not yet produced blocks using \linebreak \texttt{block\_exists\_at\_authority\_round}. For each missing leader, the function examines all blocks from the next round to count blames. Since a block blames a leader if the leader is not in the block's reference list, the function confirms that \texttt{all} references exclude the target leader.

Blames are accumulated using \texttt{StakeAggregator<IndirectQuorumThreshold>}, which enforces the $2f + 1$ requirement. The function returns true only if all missing leaders have accumulated sufficient blame, ensuring that the validator can safely proceed without waiting the full timeout.

\clearpage

\begin{lstlisting}[language=Rust, basicstyle=\footnotesize\ttfamily, breaklines=true, frame=single]
pub fn force_due_to_leader_blames(&self,
    period: u64,
) -> bool {
    let quorum_round = ...; // the next round
    if quorum_round > ... {
        let leader_round = quorum_round - 1;
        let mut leaders = self.committer.get_leaders(leader_round);
        leaders.retain(|leader| {!self
                .block_store
                .block_exists_at_authority_round(*leader, leader_round)
        });
        let quorum_blocks = self.block_store.get_blocks_by_round(quorum_round);
        for leader in leaders.into_iter() {
            let mut blame_stake_aggregator = StakeAggregator::<IndirectQuorumThreshold>::new();
            let mut reached_threshold_for_leader = false;
            for quorum_block in &quorum_blocks {
                let decider = quorum_block.reference().authority;
                if quorum_block
                    .includes().iter()
                    .all(|include| include.authority != leader)
                {
                    if blame_stake_aggregator.add(decider, &self.committee) {
                        reached_threshold_for_leader = true;
                        break;
                    }
                }
            }
            if !reached_threshold_for_leader {
                return false;
            }
        }
        true
    } else {
        false
    }
}
\end{lstlisting}

\section{Testing}\label{sec:testing}

Another advantage of building on an existing codebase was the ability to use its testing infrastructure. The simulation layer, which emulates both the \texttt{tokio} runtime and TCP networking through discrete event simulation, performed controlled, repeatable testing of wide-area network scenarios on a single machine. To accommodate \lmysticeti, the previous tests were modified to reflect its updated requirements and protocol logic, including changes to validator set-ups and DAG construction. These testing revisions were more extensive than updating the consensus logic.

Testing was conducted in four different settings to evaluate the protocol in increasingly complicated scenarios: (1) a baseline with one leader per round and no wave pipelining to establish fundamental correctness, (2) an increasing number of leaders per round from one up to the full committee size, (3) wave pipelining enabled with one leader per round, and (4) a final scenario combining multiple leaders per round with wave pipelining to verify the protocol's complete potential. These comprehensive tests validated the correctness of \lmysticeti across different modes of operation.

\subsection{\texttt{direct\_commit\_late\_call}}

The fourth testing scenario combining multiple leaders per round with wave pipelining was created specifically for \lmysticeti. \texttt{direct\_commit\_late\_call}, one of the tests, demonstrates this scenario by generating 21 rounds of blocks, before attempting consensus with wave pipelining and five leaders per round.

The test constructs a committee of six validators and configures the committer with pipelining on and five leaders per round. The \texttt{build\_dag} function generates a fully interconnected DAG across the 21 rounds. After constructing the DAG, the test calls \texttt{try\_commit} to process all accumulated blocks simultaneously. The assertion validates that the commitment sequence length matches the expected number of committed leaders across the rounds. In total, there are 100 committed leaders as five leaders are committed for 20 straight rounds. The for loop checks each committed block corresponds to the correct leader as decided by the round robin schedule from \texttt{committee.elect\_leader}, confirming that the pipelined multi-leader protocol maintains the correct leader ordering.

\begin{lstlisting}[language=Rust, basicstyle=\footnotesize\ttfamily, breaklines=true, frame=single]
fn direct_commit_late_call() {
    let committee = committee(6); // with f = 1, n = 5*1+1 = 6
    let wave_length = DEFAULT_WAVE_LENGTH;
    let number_of_leaders = committee.quorum_threshold(); // 4*1+1 = 5
    let enough_blocks = 10 * wave_length + (wave_length - 1); // 10*2+(2-1) = 21
    let mut block_writer = TestBlockWriter::new(&committee);
    build_dag(&committee, &mut block_writer, None, enough_blocks);
    let committer = UniversalCommitterBuilder::new(
        committee.clone(), block_writer.into_block_store(), test_metrics(),
    )
    .with_wave_length(wave_length)
    .with_number_of_leaders(number_of_leaders)
    .with_pipeline(true)
    .build();
    let last_committed = BlockReference::new_test(0, 0);
    let sequence = committer.try_commit(last_committed);
    assert_eq!(sequence.len(), number_of_leaders * 10 * wave_length);
    for (i, leader) in sequence.iter().enumerate() {
        if let LeaderStatus::Commit(block) = leader {
            let leader_offset = ...;
            let leader_round = ...;
            let expected = committee.elect_leader(leader_offset + leader_round);
            assert_eq!(block.author(), expected);
        } else {
            panic!(...)
        };
    }
}
\end{lstlisting}

\chapter{Evaluation}\label{chap:evaluation}

\lmysticeti was evaluated to understand how the protocol's design decisions impact real world performance. Throughout this section, \textbf{latency} refers to the duration between transaction submission by a network participant and its inclusion in a committed block, measured in milliseconds (ms), while \textbf{throughput} denotes the rate of transactions submitted to the validator network over a time period, expressed in transactions per second (tx/s).

The evaluation methodology focuses on latency measurements across varying throughput loads using Amazon Web Services (AWS). The \texttt{orchestrator} crate was used to deploy globally distributed test clusters, allowing \lmysticeti to be evaluated under realistic deployment conditions with minimal setup effort. To better support the evaluation, the \texttt{orchestrator} crate was extended to selectively target simulated validators in specified regions, improving deployment flexibility compared to the original broadcast to all model.

\section{Setup}

Validators were deployed across thirteen geographically distributed AWS regions to mimic realistic blockchain network conditions. The selected regions were Northern Virginia (us-east-1), Oregon (us-west-2), Canada (ca-central-1), Frankfurt (eu-central-1), Ireland (eu-west-1), London (eu-west-2), Paris (eu-west-3), Stockholm (eu-north-1), Mumbai (ap-south-1), Singapore (ap-southeast-1), Sydney (ap-southeast-2), Tokyo (ap-northeast-1), and Seoul (ap-northeast-2). The locations span multiple continents and time zones, introducing realistic network delays and varying connectivity patterns characteristic of globally distributed systems. Each region was provisioned with the same commodity hardware, \texttt{m5d.8xlarge} instances, providing consistent computational resources.

Client transaction generation was simulated by assigning a subset of the globally distributed instances to submit transactions at predetermined rates, creating controlled load patterns for performance analysis. All transactions were standardized at 512 bytes, with payload content remaining arbitrary since the evaluation focuses on consensus performance rather than transaction specific processing.

\section{Results}

The performance results presented in Figure~\ref{fig:latency} illustrate latency for validator networks of different sizes and fault conditions. In the 10 node and 50 node networks all validators function correctly throughout the experiment. In addition, an 11 node network with 2 crash faults was evaluated (when $f = 2$, $n = 5f + 1 = 11$). Crash faults represent validators that fail by stopping operation entirely, ceasing to send messages, propose blocks, or participate in consensus, but do not exhibit malicious behavior or send conflicting information. These crashed validators effectively become unresponsive, requiring the remaining 9 honest validators to continue consensus operations. All networks were tested under sustained throughput loads of 10,000 and 50,000 tx/s. Each configuration stabilized for several minutes before data collection commenced to ensure steady state operation. Latency measurements were then collected from a dedicated monitoring machine that was not participating in consensus for a minimum of five minutes to establish accurate metrics.

\begin{figure}[htbp]
    \centering
    \includegraphics[trim=15mm 0 15mm 0, clip, width=\textwidth]{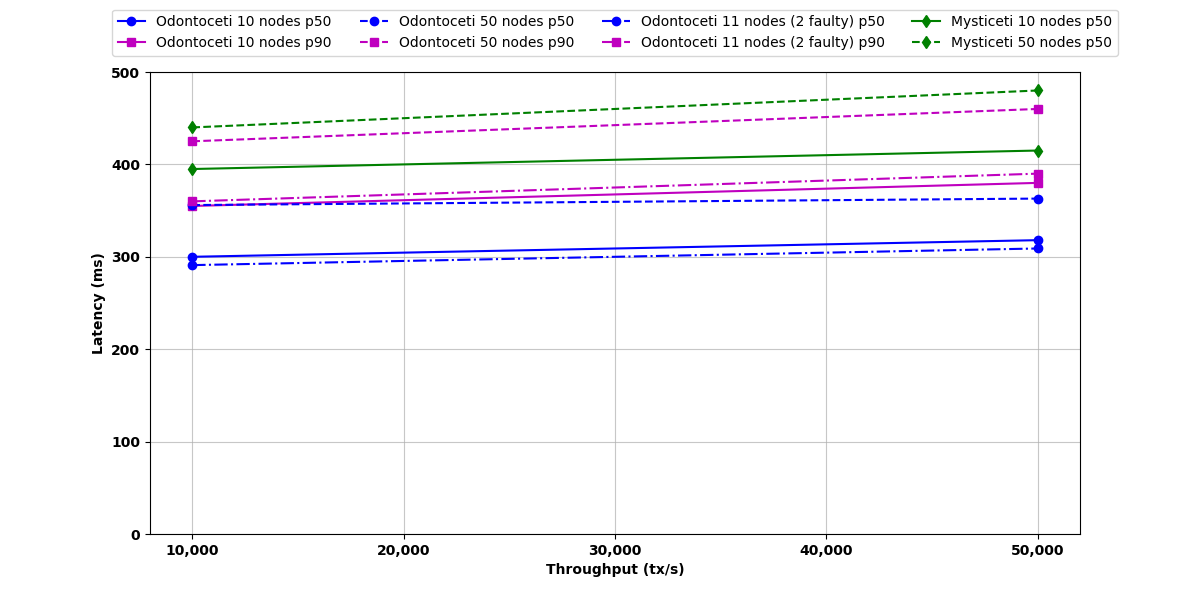}
    \caption{Graph illustrating \lmysticeti's performance under different network and throughput scenarios and in comparison against Mysticeti.}
    \label{fig:latency}
\end{figure}

\subsection{Protocol Performance}

\lmysticeti has strong speed performance across different network sizes and throughput loads in fault free conditions. For 10 node networks, the protocol has median latency of 300ms under 10,000 tx/s load and 320ms under 50,000 tx/s load, showing minimal speed change despite the five times increase in transaction volume. The 90th percentile latencies for the same configurations are 355ms and 380ms respectively, indicating reasonable tail latency behavior. Scaling to 50 node networks, \lmysticeti maintains similar performance with median latencies of 355ms and 360ms. The corresponding 90th percentile latencies are 425ms and 460ms, showing that while larger networks introduce some slow down, the performance deterioration remains manageable.

\subsection{Comparison Against Mysticeti}

To provide a meaningful performance comparison, \lmysticeti was evaluated against Mysticeti using identical experimental set ups. Both protocols were deployed using the same hardware, geographic distribution, and network conditions. The evaluation infrastructure, client workload patterns, and measurement methodology were kept consistent between both.

\lmysticeti delivers substantial latency improvements across all tested configurations compared to the Mysticeti baseline. For 10 node networks, \lmysticeti achieves 20-25\% faster median latency under both tested loads. The improvements are similarly impressive for 50 node networks, with 15-25\% faster median latencies. These performance gains demonstrate that the protocol's design choices remain effective under varying load conditions. The results validate the core thesis that reducing wave length from three rounds to two rounds yields meaningful performance benefits that outweigh the overhead introduced by the larger parent collection requirement.

\subsection{Crash Fault Resilience}

To evaluate \lmysticeti's robustness under adverse conditions, the 11 node network contained 2 crash faults. This represented realistic failure scenarios where validators become unavailable during consensus. Under these conditions, \lmysticeti exhibits median latencies which are nearly identical to the 10 node normal operation network and likely within measurement error bounds. The 90th percentile latencies are also extremely close indicating that crash faults have minimal impact on tail performance, confirming the protocol's excellent fault tolerance capabilities.

\subsection{Unsafe Variant}

To isolate the specific performance impact of \lmysticeti's decreased BFT requirement, a controlled experiment was conducted that decouples the effects of the reduced wave length from the larger parent collection threshold. This isolation is necessary for understanding which design decisions contribute to the observed performance improvements and quantifying the trade-offs involved.

The experiment compared the standard \lmysticeti implementation against an unsafe variant that mimics Mysticeti's parent collection behavior. While \lmysticeti requires references to 80\% of the network ($4f + 1$ nodes out of a $5f + 1$ network) before advancing to the next round to maintain safety guarantees (as seen in Section~\ref{sec:security}), the unsafe variant waits for only 66\% of the network, identical to Mysticeti's threshold. This comparison enables quantification of the latency overhead introduced by the stricter quorum requirement.

The results, presented in Table~\ref{tab:quorum-analysis}, reveal that the increased quorum threshold introduces a measurable but moderate latency penalty. For networks of 10 nodes under 10,000 tx/s, the standard \lmysticeti implementation exhibits 40ms additional latency compared to the unsafe variant (300ms vs. 260ms). This increases to 65ms for networks of 50 nodes under identical throughput (355ms vs. 290ms), suggesting that the impact of the larger quorum requirement scales with network size, as expected given the increased communication demands.

These findings provide insight into the performance trade-offs inherent in \lmysticeti's design. Given that \lmysticeti demonstrates a latency improvement over Mysticeti despite requiring more parents for round advancement, the results confirm that the wave length reduction from three rounds to two rounds supplies substantial performance benefits that more than compensate for the increased quorum size. Specifically, the net improvement suggests that reducing the common case commit path by one round yields latency savings that exceed the 40-65ms penalty imposed by the increased parent collection requirement.

\begin{table}[ht]
    \centering
    \begin{tabular}{|c|c|c|c|}
    \hline
    \textbf{Network Size} & \textbf{Standard \lmysticeti} & \textbf{Unsafe Variant} & \textbf{Overhead} \\
    \hline
    10 nodes & 300 & 260 & 40 \\
    \hline
    50 nodes & 355 & 290 & 65 \\
    \hline
    \end{tabular}
    \caption{Median latency (in ms) comparison between standard \lmysticeti and unsafe variant with reduced quorum threshold under 10,000 tx/s load.}
    \label{tab:quorum-analysis}
\end{table}

\chapter{Related Works}\label{chap:related}

The requirement for $n \geq 5f+1$ processes in BFT consensus protocols has changed from early theoretical foundations to modern blockchain implementations. The evolution spans three distinct phases: the foundational theoretical protocols that first used the $5f+1$ threshold, the practical challenges and improvements in implementing these fast Byzantine consensus protocols, and the recent adoption of lower fault tolerance models in current blockchains.

\section{Foundational \texorpdfstring{$5f+1$}{5f+1} Protocols}

The $5f + 1$ threshold first appeared in Ben-Or's seminal 1983 work \cite{async_random}, which demonstrated that randomized consensus was achievable in fully asynchronous networks despite the FLP impossibility result \cite{flp}, as discussed in Section~\ref{sec:flp}. Ben-Or's protocol requires $n \geq 5f + 1$ processes to tolerate $f$ Byzantine faults and achieves consensus in an expected exponential number of rounds—specifically $O(2^n)$. While this exponential complexity made the protocol impractical for large networks, it established the fundamental feasibility of asynchronous Byzantine consensus through randomization.

Concurrent to Ben-Or, Rabin \cite{rabin} presented an alternative randomized approach that also required $5f + 1$ processes. The protocol reached Byzantine agreement in a constant expected number of four communication rounds, independent of both $n$ and $f$. The protocol worked in the synchronous and asynchronous setting.

\section{Evolution of Two Round Consensus}

Beyond proving the important lower bound that two communication steps is optimal for the $5f + 1$ setting, Fast Byzantine (FaB) Paxos \cite{fab} also contributed the first protocol achieving this bound. Departing significantly from the early randomized protocols, FaB Paxos adopted the $5f + 1$ threshold as a deliberate trade-off to have faster consensus in favorable conditions. Operating in a partially synchronous network, the protocol leverages the additional processes to provide sufficient redundancy for detecting and handling Byzantine behavior through message patterns alone.

FaB Paxos uses a distinction between different types of participants: proposers who create consensus options, acceptors who vote on proposals, and learners who learn the decided values. It employs a leader-based approach where the designated leader proposes a value directly to all acceptors in the first phase. Acceptors respond with acknowledgments, and if the leader receives $4f + 1$ responses, it can immediately decide on the value without requiring a second proposal phase. This architecture reduces the typical three step consensus process to two steps by using the extra honest processes (in comparison to the $3f + 1$ network size) to ensure that honest acceptors can distinguish between proposals from honest and Byzantine leaders. However, when the ideal conditions are not met—such as during periods of asynchrony or when the leader is suspected of being faulty—the protocol falls back to a three step approach, maintaining safety while sacrificing the performance benefits of the fast path.

The FaB Paxos protocol faced practical challenges, with subsequent analysis revealing safety and liveness issues that rendered the protocol unusable in real deployments \cite{revisiting_fast_bft0}. These correctness problems were later addressed by SBFT \cite{sbft}, which resolved the identified issues while retaining the ability to commit blocks in two communication rounds, showing that the $5f + 1$ model could indeed be made practical for BFT. Additional theoretical work has further explored the $5f + 1$ model \cite{complexity_async_byz,oracle_consensus,revisiting_fast_bft1}.

\section{Modern Blockchain Adoption}

While many blockchain systems have operated under the standard $n \geq 3f + 1$ network size, there is growing interest in the community to explore lower fault tolerance thresholds to improve scalability. For example, Arc \cite{arc}, a new blockchain that currently operates with 33\% fault tolerance, explicitly states in its whitepaper that ``planned enhancements \ldots include \ldots lower fault-tolerance configurations'' as part of its goal to accomplish ``high throughput and rapid finality.'' This is emblematic of a broader shift in the blockchain space toward considering lower fault tolerances viable.

This shift toward lower fault tolerance is already being pursued by Solana \cite{solanawebsite}, which is currently developing a new consensus protocol called Alpenglow \cite{alpenglow} that operates with $n = 5f + 1$. The community is actively discussing this transition before its implementation \cite{alpenglow_code,alpenglow_forum}. Solana's motivation for adopting the $5f + 1$ model is stated as (emphasis theirs):

\begin{quote}
When discovering the fundamental result in 1980, Pease et al. considered systems where the number of nodes $n$ was small. However, today's blockchain systems consist of \textit{thousands} of nodes! While the 33\% bound also holds for large $n$, attacking one or two nodes is not the same as attacking thousands. In a large scale Proof of Stake blockchain system, running a thousand malicious ("byzantine") nodes would be a costly endeavor, as it would likely require billions of USD as staking capital. Even worse, misbehavior is often punishable, hence an attacker would lose all this staked capital. So, in a \textit{real} large scale distributed blockchain system, we will probably see \textit{significantly less} than 33\% byzantines. \cite{alpenglow}
\end{quote}

Alpenglow achieves consensus in one of two ways: blocks can be committed either through a single voting round if they receive $4f + 1$ votes, or through two voting rounds where each round requires $3f + 1$ votes. Kudzu \cite{kudzu} and Hydrangea \cite{hydrangea}, both developed concurrently within the cryptocurrency research community, employ identical commitment rules to Alpenglow, demonstrating the broader interest in this particular approach to $5f + 1$ consensus. The single round commitment is identical to \lmysticeti's direct decision rule, where all protocols commit blocks immediately upon receiving $4f + 1$ support. However, the protocols diverge in the other path. Alpenglow, Kudzu, and Hydrangea's two round commitment requires $3f + 1$ votes in both rounds, while \lmysticeti's indirect decision rule requires $2f + 1$ support in the subsequent round, with those supporters then being in an anchor's causal history. For skipping blocks, Alpenglow, Kudzu, and Hydrangea employ a symmetric approach to commits. They skip blocks that receive $3f + 1$ skip votes in either the first or second voting rounds. This is similar to \lmysticeti's direct skip rule, which requires $4f + 1$ blames, though \lmysticeti additionally supports indirect skipping through the same anchor-based mechanism used for indirect commits.

Direct performance comparisons between Alpenglow and \lmysticeti are not feasible due to different evaluation methodologies. Alpenglow's performance analysis is based on the actual stake distribution of Solana's validator network, where 65\% of stake is within 50ms round trip time of Zurich, with 25\% of the stake appearing to be located in Zurich itself. Over 80\% of stake is within 100ms round trip of Zurich, and only 1\% is farther than 200ms. This highly centralized geographic distribution creates favorable network conditions for consensus. In contrast, \lmysticeti's evaluation employs equal stake distribution among all participants in a controlled environment where round trip latencies between peers were rarely below 100ms, frequently ranged between 100 and 200ms, and occasionally spiked up to 500ms. Additionally, both Kudzu and Hydrangea remain theoretical protocols without available implementations. These contrasting evaluation approaches make useful latency comparisons between the protocols impossible.

\chapter{Conclusion}\label{chap:conclusion}

\lmysticeti is the first (practical) two round DAG-based consensus protocol. It uses a network of size $5f + 1$ and achieves a 20-25\% latency reduction over an existing state of the art DAG protocol through shortening the wave length. The early block production optimization provides additional benefits for crash faults, which are more prevalent in real blockchains than Byzantine faults. These performance gains may translate to lower transaction fees and improved user experience through faster confirmation times, addressing barriers to blockchain adoption.

The concurrent development of \lmysticeti, Alpenglow, and others may be indicative of a larger shift toward lower fault tolerance, higher performance consensus protocols. The viability of $5f + 1$ protocols for blockchain systems has been established, building a foundation for next generation high performance BFT systems. This paper demonstrates that the trade-off between fault tolerance and performance can yield practical benefits for production blockchain deployments.

\section{Future Research Directions}

Multi-threshold finality protocols could allow users to choose between faster finality at 20\% fault tolerance or stronger security at 33\% fault tolerance depending on application requirements. This approach would enable different use cases to select appropriate security-performance trade-offs while potentially supporting dynamic threshold adjustment based on network conditions.

Research is needed to formalize the economic security arguments underlying the shift from 33\% to 20\% fault tolerance in Proof of Stake blockchains. While intuitive arguments suggest that attacking these large scale networks becomes economically infeasible due to staking capital requirements and slashing penalties, more rigorous analysis of the incentive structures would give stronger evidence for deploying $5f + 1$ consensus protocols in production blockchains.

Additional optimizations may further reduce steps, waiting periods, or computation on the critical path beyond the early block production optimization.

Systems like Sui \cite{suiwebsite,mysticeti} employ a \textit{consensusless} transaction path to obtain faster finality for certain transaction types. As consensus protocols reach sufficiently fast finality, they may obviate the need for such complex alternatives. This would simplify blockchain architectures.

\appendix

\bibliographystyle{splncs04}
\bibliography{references}

\end{document}